\documentclass[a4paper,11pt,onecolumn,accepted=2022-06-23]{quantumarticle}
\pdfoutput=1
\usepackage[utf8]{inputenc}
\usepackage[english]{babel}
\usepackage[T1]{fontenc}
\usepackage{breakurl}

\usepackage{mathtools}
\usepackage{physics}    % better math symbols
\usepackage{amssymb}
\usepackage{amsfonts}
\usepackage{mathrsfs}   % cursive font
\usepackage{nicefrac}   % inline slanted fractions
\usepackage{xfrac}

\usepackage{amsthm}
\usepackage{thmtools}
\usepackage{thm-restate}

\usepackage[normalem]{ulem}

\usepackage{enumitem}

\usepackage{subcaption}

\usepackage{algpseudocode}

\usepackage{tikz}
\usetikzlibrary{decorations.markings}
\usepackage{styles/tikzit}
\usepackage{float}

\usepackage{qcircuit}

\usepackage[
  style=alphabetic,
  sorting=nyt,
  backend=biber,
  maxnames=6,
  maxcitenames=2,
  url=false,
  doi=true,
  isbn=false,
  eprint=false
]{biblatex}

\usepackage{xurl}
\usepackage{hyperref}
\hypersetup{
  colorlinks=true,
  linkcolor=cyan,
  citecolor=gray
}

\usepackage{adjustbox}

\tikzstyle{blank}=[fill=white, draw=white, shape=circle, inner sep=0, minimum size=.3cm]
\tikzstyle{black_circle}=[fill=black, draw=black, shape=circle]
\tikzstyle{white_circle}=[fill=white, draw=black, shape=circle]
\tikzstyle{blue_circle}=[fill=blue, draw=black, shape=circle]
\tikzstyle{black_square}=[fill=black, draw=black, shape=rectangle]
\tikzstyle{blue_square}=[fill=blue, draw=black, shape=rectangle]
\tikzstyle{grey square}=[fill={rgb,255: red,159; green,159; blue,159}, draw={rgb,255: red,159; green,159; blue,159}, shape=rectangle]
\tikzstyle{ctrl}=[fill=black, draw=black, shape=circle, inner sep=0, minimum size=1mm]
\tikzstyle{target inner}=[draw=black, shape=circle, cross out, rotate=45, inner sep=0, minimum size=1mm, tikzit fill={rgb,255: red,255; green,64; blue,47}]
\tikzstyle{target outer}=[draw=black, shape=circle, inner sep=0, minimum size=2mm, tikzit fill={rgb,255: red,255; green,51; blue,24}]
\tikzstyle{gate}=[fill=white, draw=black, shape=rectangle]
\tikzstyle{small_none}=[inner sep=0, minimum size=.3cm, font={\scriptsize}]

\tikzstyle{blue_dashed_arrow_right}=[->, draw=blue, dashed, thick]
\tikzstyle{black_arrow_left}=[<-]
\tikzstyle{blue_dashed_arrow_left}=[<-, draw=blue, dashed, thick]
\tikzstyle{red_dashed_arrow_left}=[draw=red, <-, dashed, thick]
\tikzstyle{red_dashed_arrow_right}=[draw=red, ->, dashed, thick]
\tikzstyle{new edge style 0}=[-, draw={rgb,255: red,144; green,144; blue,144}]
\tikzstyle{green_edge}=[->, draw={rgb,255: red,26; green,212; blue,79}, dashed, thick]
\tikzstyle{orange_edge}=[draw={rgb,255: red,255; green,176; blue,66}, ->, dashed, thick]
\tikzstyle{black_edge}=[->, thick]
\tikzstyle{new edge style 1}=[-, draw={rgb,255: red,150; green,150; blue,150}, dashed]
\tikzstyle{black_thick_l}=[thick, <-]
\tikzstyle{dotted blue}=[-, draw={rgb,255: red,93; green,195; blue,239}, dashed, thick]
\tikzstyle{new edge style 2}=[draw=red, ->]
\tikzstyle{new edge style 3}=[-, draw={rgb,255: red,230; green,54; blue,19}]

\declaretheoremstyle[
  spaceabove=5mm, spacebelow=5mm,
  headfont=\normalfont\bfseries\sffamily,
  notefont=\mdseries,
  notebraces={(}{)},
  bodyfont=\normalfont,
  postheadspace=1em,
  qed=\(\lrcorner\)
]{def}

\declaretheoremstyle[
  spaceabove=8mm, spacebelow=8mm,
  headfont=\normalfont\bfseries\sffamily,
  notefont=\mdseries,
  notebraces={(}{)},
  bodyfont=\normalfont,
  postheadspace=1em
]{thm}

\declaretheoremstyle[
  spaceabove=5mm, spacebelow=8mm,
  headfont=\normalfont\bfseries\sffamily,
  notefont=\mdseries,
  notebraces={(}{)},
  bodyfont=\itshape,
  postheadspace=1em
]{prop}

\declaretheorem[name=Definition, style=def]{Def}

\declaretheorem[name=Theorem, style=prop]{Thm}

\declaretheorem[name=Proposition, sibling=Def, style=prop]{Prop}

\declaretheorem[name=Lemma, sibling=Prop, style=prop]{Lem}

\declaretheorem[name=Corollary, numbered=no, style=prop]{Cor}

\declaretheorem[name=Remark, style=remark, numbered=no, style=def]{Rem}

\newcommand{\NN}{\mathbb{N}}
\newcommand{\ZZ}{\mathbb{Z}}

\newcommand{\RR}{\mathbb{R}}
\newcommand{\CC}{\mathbb{C}}

\newcommand{\LR}{\mathrm{L}^2(\mathbb{R})}
\newcommand{\Schw}{\Sch}

\renewcommand{\a}{\alpha}
\renewcommand{\b}{\beta}
\newcommand{\g}{\gamma}
\renewcommand{\d}{\delta}
\newcommand{\e}{\varepsilon}
\newcommand{\s}{\sigma}

\newcommand{\Ao}{\mathsf{A}}

\newcommand{\Co}{\mathsf{C}}

\newcommand{\Eo}{\mathsf{E}}
\newcommand{\Fo}{\mathsf{F}}
\newcommand{\Go}{\mathsf{G}}
\newcommand{\Ho}{\mathsf{H}}
\newcommand{\Io}{\mathsf{I}}
\newcommand{\Jo}{\mathsf{J}}
\newcommand{\Ko}{\mathsf{K}}

\newcommand{\Mo}{\mathsf{M}}

\newcommand{\Po}{\mathsf{P}}
\newcommand{\Qo}{\mathsf{Q}}

\newcommand{\So}{\mathsf{S}}
\newcommand{\To}{\mathsf{T}}
\newcommand{\Uo}{\mathsf{U}}

\newcommand{\Wo}{\mathsf{W}}
\newcommand{\Xo}{\mathsf{X}}

\newcommand{\Zo}{\mathsf{Z}}

\newcommand{\Ach}{\mathscr{A}}
\newcommand{\Bch}{\mathscr{B}}
\newcommand{\Cch}{\mathscr{C}}

\newcommand{\Fch}{\mathscr{F}}
\newcommand{\Gch}{\mathscr{G}}
\newcommand{\Hch}{\mathscr{H}}

\newcommand{\Och}{\mathscr{O}}

\newcommand{\Sch}{\mathscr{S}}
\newcommand{\Tch}{\mathscr{T}}
\newcommand{\Uch}{\mathscr{U}}
\newcommand{\Vch}{\mathscr{V}}

\newcommand{\Zch}{\mathscr{Z}}

\newcommand{\Hcal}{\mathcal{H}}

\DeclareMathOperator{\CXo}{\Co\Xo}
\DeclareMathOperator{\CZo}{\Co\Zo }

\begin{document}

\title{Flow conditions for continuous-variable measurement based quantum computing}

\author{Robert I. Booth}
\affiliation{
  Sorbonne Universit\'e, CNRS, LIP6,
  4 place Jussieu, \mbox{F-75005} Paris, France
}
\affiliation{
  LORIA CNRS, Inria Mocqua, Universit\'e de Lorraine, \mbox{F-54000} Nancy,
  France
}
\author{Damian Markham}
\affiliation{
  Sorbonne Universit\'e, CNRS, LIP6,
  4 place Jussieu, \mbox{F-75005} Paris, France
}
\affiliation{
  JFLI, CNRS / National Institute of Informatics, University of Tokyo, Tokyo, Japan
}

\maketitle

\begin{abstract}
  In measurement-based quantum computing (MBQC), computation is carried out by a
  sequence of measurements and corrections on an entangled state. Flow, and
  related concepts, are powerful techniques for characterising the dependence of
  the corrections on previous measurement results. We introduce flow-based
  methods for quantum computation with continuous-variable graph states, which
  we call CV-flow. These are inspired by, but not equivalent to, the notions of
  causal flow and g-flow for qubit MBQC. We also show that an MBQC with CV-flow
  approximates a unitary arbitrarily well in the infinite-squeezing limit,
  addressing issues of convergence which are unavoidable in the
  infinite-dimensional setting.
  %Like for discrete variables, these constructions are useful for determining when an arbitrary CV graph state can be used for a practical computation and investigating the trade-off between classical and quantum depth, which has led to depth complexity separation between MBQC and circuit based quantum  computing.
  In developing our proofs, we provide a method for converting a CV-MBQC
  computation into a circuit form, analogous to the circuit extraction method of
  Miyazaki et al, and an efficient algorithm for finding CV-flow when it exists
  based on the qubit version by Mhalla and Perdrix. Our results and techniques
  naturally extend to the cases of MBQC for quantum computation with qudits of
  prime local dimension.
\end{abstract}

Causal flow is a graph-theoretical tool for characterising the quantum states
used in measurement-based quantum computation (MBQC) and closely related to the
measurement calculus \cite{raussendorf_one-way_2001,
  raussendorf_computational_2002, danos_determinism_2006}.
Its original purpose was to identify a class of qubit graph states that can be
used to perform a deterministic MBQC despite inherent randomness in the outcomes
of measurements, but it has since found applications to a wide variety of
problems in quantum information theory.

Along with its generalisation g-flow \cite{danos_measurement_2007}, causal flow has
been used to parallelise quantum circuits by translating them to MBQC
\cite{broadbent_parallelizing_2009}, to construct schemes for the verification of
blind quantum computation \cite{fitzsimons_unconditionally_2017, mantri_flow_2017},
to extract bounds on the classical simulatability of MBQC
\cite{markham_entanglement_2014}, to prove depth complexity separations between the
circuit and measurement-based models of computation
\cite{broadbent_parallelizing_2009, miyazaki_analysis_2015}, and to study
trade-offs in adiabatic quantum computation \cite{antonio_adiabatic_2014}.
A relaxation of these notions was also used in \cite{mhalla_which_2014} to further
classify which graph states can be used for MBQC.
g-flow can also be viewed as a method for turning protocols with post-selection
on the outcomes of measurements into deterministic protocols without
post-selection, which has been useful for applying ZX-calculus techniques
\cite{backens_there_2021}.
This perspective has been used for the verification of measurement-based quantum
computations \cite{hutchison_rewriting_2010}, as well as state of the art quantum
circuit optimisation techniques \cite{duncan_graph-theoretic_2020} and even to
design new models of quantum computation \cite{de_beaudrap_pauli_2020}.

Concurrently, it has become apparent that quantum computing paradigms other than
qubit based models might offer viable alternatives for constructing a
quantum computer.
Continuous-variable (CV) quantum computation, which has a physical interpretation
as interacting modes of the quantum electromagnetic field, is such a non-standard
model for quantum computation \cite{lloyd_quantum_1999, braunstein_quantum_2005}.
CV computation also has implementational advantages over the discrete variable
case. To the author's knowledge, the largest entangled states observed
experimentally to date (in terms of number of involved systems) remain those
obtained in CV optical experiments using time-multiplexing, beating their
discrete counterparts by many orders of magnitude
\cite{yokoyama_ultra-large-scale_2013, yoshikawa_generation_2016,
asavanant_time-domain_2019}. That said, there of course remain challenges to a CV
quantum computer. While it is relatively easier to generate entanglement in CV,
there are other operations necessary for quantum speed-up, dubbed
\emph{non-Gaussian} operations, which are much more difficult to obtain since
they correspond to interactions typically observed only at very high energies or
in very specific states of matter. Some progress has recently been made on this
aspect \cite{miyata_implementation_2016, konno_nonlinear_2021,
konno_non-clifford_2021}. Secondly, genuinely CV quantum error correction codes
have proved elusive. This is partially because it has been shown that no
satisfactory CV code is possible without non-Gaussian operations
\cite{eisert_distilling_2002, niset_no-go_2009, vuillot_quantum_2019}, but even
when assuming access to such operations, results beginning to tackle natural
classes of errors have only recently been obtained \cite{noh_encoding_2020,
  hao_topological_2021}.

The MBQC framework has been extended to the CV case, with a surprisingly similar
semantics \cite{zhang_continuous-variable_2006, menicucci_universal_2006}.
Accordingly, some structures transfer naturally from DV to CV, and it is of
interest to investigate if it is possible to define notions of flow for CV-MBQC.
That one should be interested in CV-MBQC with arbitrary entanglement topologies,
as we investigate here, is justified by results such as those of
\cite{alexander_flexible_2016}, in which it is shown that various advantages can
be obtained over the standard square-lattice CV cluster states.
However, CV-MBQC comes with an additional complication: the gate teleportation
protocol is an approximation to the desired unitary gate.
This comes about because the unitary can only be understood to be obtained in the
limit of infinite squeezing of a physical protocol.
The convergence of this approximation is implicit in CV teleportation protocols,
but the convergence of an MBQC with arbitrary entanglement topologies is not
assured.

In this paper, we define such a notion, converting the results on flow and 
g-flow from \cite{browne_generalized_2007} to continuous variables, and use it 
to identify a class of graph states that can be used for convergent MBQC 
protocols. In section \ref{sec:prelims}, we review CV quantum computation and 
define our computational model. In section \ref{sec:flow}, we state our CV flow 
conditions, and prove that they give rise to a suitable MBQC protocol with 
auxiliary squeezed states \cite{lvovsky_squeezed_2015}. In section 
\ref{sec:circuit}, we construct a quantum circuit extraction scheme for our 
flow-based CV-MBQC protocol, proving convergence in the infinite-squeezing 
limit. In section \ref{sec:qudit}, we briefly explain how our techniques adapt 
to the qudit case, yielding an analogous MBQC scheme and a corresponding 
circuit extraction. Finally, appendix \ref{app:convergence} contains some of
the more technical proofs, appendix \ref{app:algorithm} a polynomial-time
algorithm for determining if a graph has CV-flow, appendix \ref{app:comparison}
a comparison of our CV flow conditions to the original DV conditions when this
makes sense, and appendix \ref{app:depth} an example of depth-complexity
advantage using CV-MBQC compared to a circuit acting on the same number of
inputs.

\paragraph{Acknowledgements}
The authors were supported by the ANR VanQuTe project (ANR-17-CE24-0035). RIB is
now based at the University of Edinburgh, but did all the work at his previous
institutions. We thank Simon Perdrix, Elham Kashefi, Ulysse Chabaud and
Francesco Arzani for enlightening discussions.

\section{Preliminaries}
\label{sec:prelims}
Our model is based on measurement-based quantum computation using
continuous variable graph states (CV-MBQC), as described in
\cite{zhang_continuous-variable_2006,menicucci_universal_2006,gu_quantum_2009}.
We first recall some background relevant for continuous-variable quantum
computation, and briefly review CV-MBQC, which is the primary motivation for
deriving flow conditions for continuous variables.
We then introduce open graphs and how they relate to the states used in CV-MBQC.

\paragraph{\sffamily Notation.}
We use \(\abs{X}\) to denote the cardinality of the set \(X\).

Sans-serif font will denote linear operators on a Hilbert space:
\(\mathsf{A, B, \dots, X, Y, Z}\), and \(\mathsf{A^*}\) the Hermitian
adjoint of \(\mathsf{A}\).
\(\Io\) is the identity operator.
Cursive fonts are used for completely positive trace non-increasing maps
(quantum channels): \(\Ach, \Bch, \dots\)

\subsection{Computational model}
\label{ssec:cv-comp}

In CV quantum computation, the basic building block is the
\textbf{qumode}\footnote{This terminology comes from quantum optics, where we can
  identify each quantisation mode of the quantum electromagnetic field with a
  space \(\mathrm{L}^2(\RR)\) \cite{fabre_modes_2020}.}, a complex, countably
infinite-dimensional, separable Hilbert space \(\Hcal = \mathrm{L}^2(\RR,\CC)\)
which takes the place of the qubit. \(\Hcal\) is a space of square integrable
complex valued functions: an element of \(\phi \in \Hcal\) is a function \(\RR
\to \CC\) such that
\begin{equation}
  \int_{\RR} \abs{\phi(x)}^2 \dd{x} < \infty,
\end{equation}
and where the Hilbert inner product is
\begin{equation}
  \langle \psi, \phi \rangle \coloneqq \int_{\RR} \bar{\psi}(x) \phi(x) \dd{x},
\end{equation}
with corresponding norm \(\norm{\psi}_\Hcal \coloneqq \sqrt{\langle \psi,\psi
  \rangle}.\)

Each qumode is equipped with a pair of unbounded linear position and momentum
operators \(\Qo\) and \(\Po\), which are defined on the dense subspace \(\Schw
\subseteq \Hcal\) of Schwartz functions\footnote{This is a technical condition
  which ensures that for any real polynomial \(\mathsf{p}\) in \(\Qo\)
  and \(\Po\), \(\mathsf{p} \phi\) remains a square-integrable function,
  something which is not true in general.
  The natural setting for this discussion is the \emph{rigged Hilbert space}, or
  Gel'fand triple, \(\Schw \subseteq \Hcal \subseteq \Schw^*,\) where \(\Schw^*\)
  is the continuous dual of Schwartz space or \emph{space of tempered
    distributions}.
  We refer the interested reader to the classic series of books by
  Gel'fand et al. \cite{gelfand_generalized_2016, gelfand_generalized_2016-1,
    gelfand_generalized_2016-2, gelfand_generalized_2016-3}, and also to
  \cite{gadella_unified_2002, celeghini_groups_2019} and the references therein
  for a recent discussion of its application to quantum mechanics.
}, along with any inhomogeneous polynomial thereof:
\begin{equation}
  \qq{for any \(\phi \in \Schw\),} \Qo \phi(x) \coloneqq x \phi(x) \qand \Po
  \phi(x) \coloneqq -i\dv{\phi(x)}{x}.
\end{equation}
From these, we can define the corresponding translation operators (continuously
extendable to all \(\Hcal\)):
\begin{equation}
  \qq{for any \(s \in \RR\),} \Xo(s) \coloneqq \exp( -is \Po ) \qq{and} \Zo(s)
  \coloneqq \exp( is \Qo ),
\end{equation}
such that
\begin{align}
  \Xo(-s) \Qo \Xo(s) &= \Qo + s\Io; \\
  \Zo(-s) \Po \Zo(s) &= \Po + s\Io.
\end{align}
In fact, all four of these operators are defined by the exponential Weyl
commutation relations (up to unitary equivalence, by the Stone-von Neumann
theorem, see \cite{hall_quantum_2013} section 14):
\begin{equation}
  \qq{for any \(s,t \in \RR\),} \Xo(s) \Zo(t) = e^{ist} \Zo(t) \Xo(s),
\end{equation}
which generalise the canonical commutation relations, and further related by the
Fourier transform operator \(\mathsf{F} : \Hcal \to \Hcal\):
\begin{equation}
  \mathsf{FQF^*} =\Po
  \qq{and}
  \mathsf{FPF^*} = \mathsf{-Q}.
\end{equation}

The \textbf{squeeze operator} is defined for any real number \(\eta > 0\),
called the \textbf{squeezing factor}, by
\begin{equation}
  \So(\eta) \coloneqq \exp\big(-i \ln(\eta) (\Qo\Po + \Po\Qo)\big).
\end{equation}
Then \(\So(\eta)\psi(x) = \sqrt{\eta^{-1}} \psi(\eta^{-1} x)\) so that:
\begin{align}
  \So(\eta)^* \Qo \So(\eta) = \eta \Qo, &\quad  \So(\eta)^* \Po \So(\eta) = \eta^{-1} \Po \\
  \So(\eta)^* \Zo(s) \So(\eta) = \Zo(\eta s), &\quad  \So(\eta)^* \Xo(s) \So(\eta) = \Xo(\eta^{-1} s).
\end{align}

Following Lloyd and Braunstein \cite{lloyd_quantum_1999,
  braunstein_quantum_2005}, the state of a set of \(N\) qumodes can be used to
encode information and perform computations just as one would with a register of
qubits, using unitaries from the set
\begin{equation}
  \{\Fo, \exp(is\Qo_j), \exp(is\Qo_j^2), \exp(is\Qo_j^3), \exp(i s \Qo_j \Qo_k)
    \mid s \in \RR, j,k \in \{1,...,N\}\},
\end{equation}
and where states are obtained with the usual tensor product of Hilbert spaces.
The indices \(j,k\) indicate on which subsystems in the tensor product the
operators act.
For brevity and by analogy with DV, we write:
\begin{align}
  \Co\Zo_{j, k}(s) &\coloneqq \exp(i s \Qo_j \Qo_k), \label{eq:CZ} \\
  \Co\Xo_{j, k}(s) &\coloneqq \exp(i s \Qo_j \Po_k) = \Fo_k \CZo_{j,k}(s) \; \Fo_k^*, \label{eq:unitary} \\
  \Uo_k(\a,\b,\g) &\coloneqq \exp( i \a \Qo_k ) \exp( i \b \Qo_k^2 )
                    \exp( i \g \Qo_k^3 ). \label{eq:unitaryb}
\end{align}

This model of computation is strong enough to encode qubit quantum computation
\cite{gottesman_encoding_2001}, and is universal in the sense that any unitary can be approximated by combinatons of applications of (\ref{eq:CZ}) - (\ref{eq:unitaryb})
\cite{lloyd_quantum_1999}.

\subsubsection*{Mixed states}

Even ignoring inevitable experimental noise, since the teleportation procedures that we consider are not unitary in general,
but only in an ideal limit, we need to work with mixed states.
In continuous variables, there are further mathematical technicalities involved
with defining density operators, which we deal with here
\cite{shirokov_approximation_2008, hall_quantum_2013}.

Let \(\Hcal\) be a separable Hilbert space, and \(\mathfrak{B}(\Hcal)\) be the
algebra of bounded operators on \(\Hcal\).
We say that a self-adjoint, positive operator \(\Ao \in \mathfrak{B}(\Hcal)\) is
\textbf{trace-class} if for an arbitrary choice of basis \(\left\{ e_i
\right\}\) of \(\Hcal\), we have
\begin{equation}
  \sum_{n \in \NN} \langle e_n, \Ao e_n \rangle < +\infty.
\end{equation}
An operator \(\Ao \in \mathfrak{B}(\Hcal)\) is itself trace-class if
the positive self-adjoint operator \(\sqrt{\Ao^* \Ao}\), defined using the
functional spectral calculus, is trace-class.

The set of trace-class operators forms a Banach algebra
\(\mathfrak{T}(\Hcal)\) with norm given by the trace:
\begin{equation}
  \qq{for any \(\Ao \in D(\Hcal),\)} \tr(\Ao) \coloneqq \sum_{n \in \NN} \langle e_n, \sqrt{\Ao^* \Ao} e_n \rangle,
\end{equation}
and a positive self-adjoint operator \(\rho \in \mathfrak{T}(\Hcal)\) is called
a \textbf{density operator} if \(\tr(\rho) = 1\). We denote \(D(\Hcal)\) the set
of density operators, and for any state \(\psi \in \Hcal\), the projector
\(\rho_\psi : \phi \mapsto \langle \psi, \phi \rangle \psi\) is a density
operator since \(\tr(\rho_\psi) = \norm{\psi}^2\). The set \(D(\Hcal)\) thus
corresponds to a set of quantum states which extends the space \(\Hcal\). If
a density operator takes the form \(\rho_\psi\) for some \(\psi \in \Hcal\), we
say it is a \textbf{pure state}, otherwise it is a \textbf{mixed state}.
% It follows from the spectral theorem that every mixed state can be interpreted
% as a probability distribution over a countable set of pure states, represented
% in \(D(\Hcal)\) as a convex sum \(\rho = \sum_j p_j \rho_{\psi_j}\).

However, some of the proofs of convergence we use will depend on stronger 
assumptions on the set of input states which are allowed. We will need to make
use of the fact that the Wigner function of inputs to a quantum teleportation
are Schwartz functions. A theory of these density operators, Schwartz density
operators, was developed in \cite{keyl_schwartz_2016}, but since all reasonable
physical states are included in this set \cite{bohm_dirac_1989}, we will simply
refer to them as \textbf{physical states}.
 
\subsubsection*{Topologies on the set of quantum operations}

A linear map \(\Vch : \mathfrak{T}(\Hcal) \to \mathfrak{T}(\Hcal)\) is
trace non-increasing if for any \(\rho \in D(\Hcal)\),
\(\tr(\Vch[\rho]) \leqslant 1\), and completely positive if the dual map
is completely positive.\footnote{The dual map of \(\Vch\) is the map \(\Vch^* :
  \mathfrak{B}(\Hcal) \to \mathfrak{B}(\Hcal)\) given by \(\tr(\Ao
  \Vch[\rho]) = \tr(\Vch^*[\Ao] \rho)\) for any \(\rho \in D(\Hcal)\).
  It is completely positive if for all \(n \in \NN\), \(\Vch^* \otimes id_n :
  \mathfrak{B}(\Hcal) \otimes \CC^n \to \mathfrak{B}(\Hcal) \otimes
  \CC^n\) is positive.}
Then, a completely positive, trace non-increasing linear map
\(\mathfrak{T}(\Hcal) \to \mathfrak{T}(\Hcal)\) implements a
physical transformation on the set of states \(D(\Hcal)\), called a
\textbf{quantum operation}. For example, for any unitary \(\Uo\) acting on
\(\mathscr{H}\), there is a quantum operations \(\Uch_\Uo : \mathfrak{T}(\Hcal)
\to \mathfrak{T}(\Hcal)\) given by conjugation by \(\Uo\), i.e. \(\Uch_\Uo[\rho]
= \Uo \rho \Uo^*\).

The set of all quantum operations \(\mathfrak{T}(\Hcal) \to
\mathfrak{T}(\Hcal)\) can be given several different topologies.
The simplest is the \textbf{uniform topology}, given by the norm
\begin{equation}
  \norm{\Vch} = \sup_{\rho \in D(\Hcal)} \tr(\Vch[\rho]).
\end{equation}
However, as discussed in \cite{shirokov_approximation_2008, wilde_strong_2018,
  pirandola_teleportation_2018}, the uniform topology is inappropriate for
considering the approximation of arbitrary quantum operations in
infinite-dimensional Hilbert spaces.
Instead, we use a coarser topology, the \textbf{strong topology}, which is
generated by the family of semi-norms:
\begin{equation}
  \qq{for each \(\rho \in D(\Hcal)\),} \Vch \longmapsto \tr(\Vch[\rho]),
\end{equation}
and a sequence of quantum operations \((\Vch_k)_{k \in \NN}\) converges to
\(\Vch\) in the strong topology if and only if for every \(\rho \in
D(\Hcal)\),
\begin{equation}
  \lim_{k \to \infty} \Vch_k[\rho] = \Vch[\rho].
\end{equation}

Thus, the sequence \((\Vch_k)_{k \in \NN}\) can be viewed as a pointwise
approximation to \(\Vch\), and this is the perspective we take in this paper: we
construct an MBQC procedure, associated to a flow condition, which converges
strongly to unitary quantum circuits in the ideal limit of the approximation.

\subsection{MBQC}
\label{ssec:cv-mbqc}

The workhorse of MBQC (in DV and CV) is gate teleportation, which makes it
possible to apply a unitary operation from a specific set on a qumode by
entangling it with another qumode and measuring.
Informally, for an input state \(\phi \in \Hcal\) the idealised quantum circuit
for gate teleportation in CV (assuming infinite squeezing, see below) is:
\begin{center}
  \leavevmode
  \centering
  \Qcircuit @C=1.2em @R=0.4em {
    \lstick{\phi} & \ctrl{3} & \gate{\mathsf{U}}
    & \measureD{\Po} & \rstick{m} \cw \\
    & & \hspace{-1.8cm}w \\
    \\
    \lstick{\delta_\Po} & \control \qw & \qw
    & \qw & \rstick{\Xo(wm) \So(w) \Fo \Uo \phi} \qw \\ 
  }
  \vspace{3mm}
\end{center}
The auxiliary input \(\delta_\Po\) is a momentum eigenstate with eigenvalue
\(0\), or Dirac delta distribution centered at \(x = 0\).
The two qumode interaction is \(\Co\Zo_{12}(w)\) (equation \eqref{eq:CZ}),
\(\mathsf{U}\) is any unitary gate that commutes with \(\Co\Zo_{12}(w)\) (such
as the unitary \(\Uo(\a,\b,\g)\) from equation \eqref{eq:unitary}), and we
measure the first qumode in the \(\Po\) basis.
If we view \(\mathsf{U}\) as a change of basis for the measurement, this
``gadget'' allows us to perform universal computation using only entanglement
and measurements, in the sense of Lloyd and Braunstein
\cite{menicucci_universal_2006, kalajdzievski_exact_2021}.
However, there is an extra gate \(\Xo(w \cdot m)\) on the output of the
computation \emph{which depends on the result of the measurement \(m\)}.
We call this the \textbf{measurement error}. In the course of a computation, it is the role of flow to describe how to correct for these measurement errors.

The representation of gate teleportation above is an idealisation, and
necessarily only approximates the limit of physically achievable processes.
For our needs, we develop this in detail now.
Formally, Dirac deltas are tempered distributions in the Schwartz sense but
cannot be interpreted as input states (even in principle) since the space of
distributions \(\Sch^*\) is much larger than the state space \(\Hcal\).
It is necessary to use an approximation, given by the following parametrised
Gaussian states, called \textbf{squeezed states}. Let \(g_1\) be an
\(\mathscr{L}^2\)-normalised Gaussian distribution on \(\RR\), explicitly given
by
\begin{equation}
  g_1(x) \coloneqq \frac{1}{\sqrt[4]{2\pi}} e^{-\frac{x^2}{2}},
\end{equation}
and put \(g_\eta \coloneqq S(\eta)g_1\).\footnote{From a physical perspective,
  \(g_1\) is the vacuum state uniquely defined by \((\Qo + i\Po) g_0 = 0\).} In
other words,
\begin{equation}
  g_\eta(x) = \frac{1}{\sqrt[4]{2\pi\eta^2}} e^{-\frac{x^2}{2\eta^2}}.
\end{equation}
In the limit \(\eta \to +\infty\), this state will play the role of the
auxiliary state for the teleportation, but as per the previous discussion,
\(\lim_{\eta \to +\infty} g_\eta \notin \Hcal\),
i.e. the limit is divergent in state space. One might think of this
infinitely-squeezed limit as a momentum eigenstate, but it is not formally a
member of the Hilbert space and should not be treated as such. In particular, it
does not behave well under measurement since this would correspond to taking the
inner product of two distributions, which is not well-defined in the Schwartz
theory. However, as we shall see, the teleportation map \emph{does} converge to
a unitary acting on \(\Hcal\).

The circuit for the gate teleportation procedure with input \(\rho \in
D(\Hcal)\) and denoting \(\s_\eta\) the density operator of a squeezed state
\(g_n\) (as discussed in section \ref{ssec:cv-comp}) is
\begin{equation}
  \label{eq:teleportation_circuit}
  \Qcircuit @C=1.2em @R=0.4em {
    \lstick{\rho} & \ctrl{2} & \gate{\Uo(\a,\b,\g)}
    & \measureD{\Po} & \controlo{2} \cw & \\
    & \hspace{-4mm}w & & & \cwx \\
    \lstick{\s_\eta} & \control \qw & \qw & \qw & \gate{\Xo(-wm)} \cwx &
    \rstick{\Tch_\eta(\a,\b,\g)[\rho]} \qw \\ 
  }
\end{equation}
where the output is given by a quantum channel \(\Tch_\eta(\a,\b,\g)\) which
is not unitary in general.

In order to give an explicit form for \(\Tch_\eta(\a,\b,\g)\), we resort to the
following ``trick'': we identify the circuits
\begin{equation}
  \label{eq:semantics_trick}
  \adjincludegraphics[height=1.8cm, valign=c]{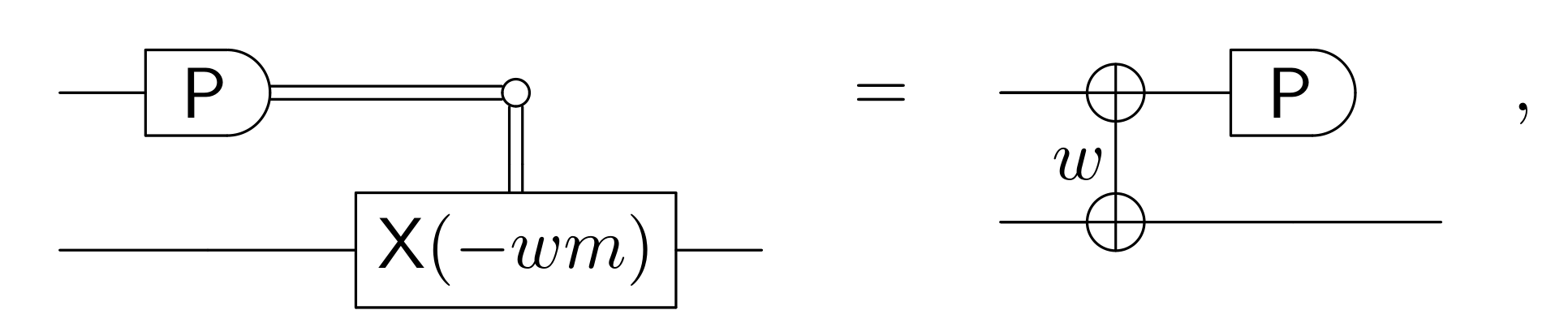}
\end{equation}
% \begin{equation}
%   \label{eq:semantics_trick}
%   \Qcircuit @C=1.2em @R=0.4em {
%     & \measureD{\Po} & \controlo{2} \cw & \\
%     & & \cwx \\
%     & \qw & \gate{\Xo(-wm)} \cwx & \qw 
%   }
%   \quad =\quad
%   \Qcircuit @C=1.2em @R=0.4em {
%     & \targ \ctrl{2} & \measureD{\Po}  \\
%     & \hspace{-4mm} w & \\
%     & \targ \qw & \qw & \qw
%   }
%   \quad,
% \end{equation}
where in both circuits the outcome of the measurement is discarded. The
two-sided control gate in the RHS represents the unitary gate
\(\exp(iw\Po_1\Po_2)\). That these two circuits are equal is intuitively
understood from the finite-dimensional case, and can be verified by explicit
calculation, for example on a basis of \(\mathcal{H}\). Then, the output state
of the RHS of this equation can be expressed using the partial trace:
\begin{equation}
  \rho_{\mathrm{out}}
  = \tr_1(\exp(iw\Po_1\Po_2) \rho_{\mathrm{in}} \exp(-iw\Po_1\Po_2))
  = \tr_1(\mathscr{U}_{\exp(iw\Po_1\Po_2)}[\rho_{\mathrm{in}}]),
\end{equation}
where \(\mathscr{U}_\Uo\) is the quantum channel corresponding to conjugation by the
unitary \(\Uo\).

It follows that the quantum channel implemented by the quantum gate
teleportation protocol (equation~\eqref{eq:teleportation_circuit}) is, for any
input state \(\rho \in D(\mathcal{H})\),
\begin{equation}
  \Tch_\eta(\a,\b,\g,w)[\rho]
  = \tr_1 \circ\, \mathscr{U}_{\exp(iw\Po_1\Po_2)} \circ \mathscr{U}_{\Uo_1(\a,\b,\g)} \circ \mathscr{U}_{\CZo_{1,2}(w)} [\rho \otimes \sigma_\eta].
\end{equation}

\begin{restatable}[Teleportation convergence]{Prop}{GateTele}
  \label{prop:tele-circuit}
  For any \(\a, \b, \g, w \in \RR\) and any \(\rho \in D(\Hcal)\),
  \begin{equation}
    \lim_{\eta \to \infty}
    \Tch_\eta(\a,\b,\g,w)[\rho]
    = \mathscr{U}_{\So(w) \Fo \Uo(\a,\b,\g)}[\rho],
  \end{equation}
  where \(\mathscr{U}_\Uo\) is the quantum channel corresponding to conjugation by the
  unitary \(\Uo\).
\end{restatable}
\begin{proof}
  The proof is left to appendix \ref{sapp:tele-circuit}.
\end{proof}

As discussed previously, this is equivalent to convergence in the strong
topology on the set of quantum channels, which is a weaker notion than uniform
convergence \cite{shirokov_approximation_2008}.
This is the best we can expect for quantum teleportation procedures in CV in
general (see \cite{wilde_strong_2018, pirandola_teleportation_2018} for a
discussion).
However, if we restrict the set of states to a compact subset, such as when the
total energy is bounded, this result can be strengthened to uniform convergence
\cite{sharma_characterizing_2020}.

The question of convergence in the infinite squeezing limit has
been considered for some slightly different but related protocols, such as the
convergence of the quantum state teleportation protocol of
\cite{wilde_strong_2018, sharma_characterizing_2020}.\footnote{The quantum state
  teleportation protocol should not be confused with the gate teleportation we use in
  this article. The state teleportation involves three qumodes: an initial qumode is
  teleported using an auxiliary two qumode state and entangling two-qumode
  measurement. In contrast, the gate teleportation only involves two qumodes.}
Their result does not immediately apply in our case, but we use some of their
ideas as well as some standard results of functional analysis in formulating our
proof.

\subsection{Graph states}

\begin{figure}
  \centering
  \begin{subfigure}[b]{0.25\textwidth}
    \centering
    \begin{tikzpicture}
	\begin{pgfonlayer}{nodelayer}
		\node [style={white_circle}] (0) at (1, 1) {};
		\node [style={black_circle}] (1) at (-1, 1) {};
		\node [style={black_circle}] (2) at (-1, -1) {};
		\node [style={black_square}] (3) at (1, -1) {};
		\node [style=blank] (5) at (0, 1) {\scriptsize $a$};
		\node [style=blank] (6) at (-1, 0) {\scriptsize $b$};
		\node [style=blank] (7) at (0, -1) {\scriptsize $d$};
		\node [style=blank] (8) at (0, 0) {\scriptsize $c$};
	\end{pgfonlayer}
	\begin{pgfonlayer}{edgelayer}
		\draw (3) to (1);
		\draw (1) to (0);
		\draw (2) to (3);
		\draw (2) to (1);
	\end{pgfonlayer}
\end{tikzpicture}
    \caption{Open graph}
  \end{subfigure}
  \begin{subfigure}[b]{0.65\textwidth}
    \centering
    \(\Go_\eta \phi = \underbrace{\Co\Zo_{1,2}(a) \Co\Zo_{1,3}(b)
      \Co\Zo_{1,4}(c) \Co\Zo_{3,4}(d)}_{\Eo_G} (g_\eta \otimes g_\eta \otimes
    g_\eta \otimes \phi)\)
    \caption{Graph state}
  \end{subfigure}
  \caption{Example of an open graph (a) and the associated graph state (b). Black vertices
    are to be measured, white vertices are outputs, and the square vertex represents
    an arbitrary input \(\phi \in \Hcal\).}
  \label{fig:graph_state}
\end{figure}
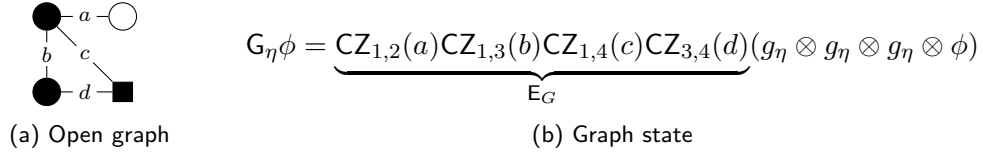

An \(\RR\)-edge-weighted graph \(G\) is a pair \((V, A)\) consisting of
a set \(V\) of vertices and a symmetric matrix \(A \in \RR^{\abs{V} \times
  \abs{V}}\), the \textbf{adjacency matrix} of \(G\), which identifies the
weight of each edge.
Furthermore: if \(A_{j, k} = 0\) then there is no edge between \(j\) and \(k\);
and for any \(j \in G\), \(A_{j,j} = 0.\) If \(j \in V\) we write \(N(j)
\coloneqq \{k \in V \mid A_{j,k} \neq 0\}\) the set of neighbours of \(j\) in
\(G\), \emph{excluding \(j\) itself}.

A (CV) \textbf{open graph} \((G,I,O)\) is an undirected \(\RR\)-edge-weighted
graph \(G=(V,A)\), along with two subsets \(I\) and \(O\) of \(V\), which
correspond to the inputs and outputs of a computation.
To this abstract graph, we associate a physical resource state, the \textbf{graph
  state}, to be used in a computation: each vertex \(j\) of the graph corresponds
to a single qumode and thus to a single pair \(\{\Qo_j, \Po_j\}\)
\cite{zhang_continuous-variable_2006,menicucci_universal_2006}.
This graphical notation is also somewhat similar to later works
for Gaussian states \cite{zhang_graphical_2008, zhang_graphical_2010, 
menicucci_graphical_2011}, although with a different focus: our notation only 
represents a small subset of the full set of multimode Gaussian states, but is 
used to reason about non-Gaussian operations on that state.

For a given input state \(\psi\) on \(\abs{I}\) modes, the graph state can
be constructed as follows:
\begin{enumerate}
\item Initialise each non-input qumode, \(j \in I^{\mathsf{c}}\), in the
  squeezed state \(g_\eta\), resulting in a separable state of the form
  \(g_\eta^{\otimes \abs{I^{\mathsf{c}}}} \otimes \psi\).
\item For each edge in the graph between vertices \(j\) and \(k\) with
  weight \(A_{j, k} \in \RR\), apply the entangling operation
  \(\Co\Zo_{j,k}\big(A_{j, k}\big)\) between the corresponding qumodes.
\end{enumerate}
We denote \(\Eo_G\) the product of the entanglement operators used to construct
the graph state (since these all commute no caution is needed with the order of
the product).

For a given open graph \((G,I,O)\), and input state \(\rho \in D(\Hcal^{\otimes
  \abs{I}})\), we write the corresponding ``graph'' state:
\begin{equation}
  \label{eq:g_eta_def}
  \Gch_\eta[\rho] = \Uch_{\Eo_G} (\rho \bigotimes_{j \in I^\mathsf{c}} \sigma_\eta).
\end{equation}
\(\Gch_\eta\) is a quantum channel as the composition of an isometry and a
unitary. For a pure input state \(\phi \in \Hcal^{\otimes \abs{I}}\), we of
course have the corresponding isometry \(\Go_\eta \phi = \Eo_G (\phi
\bigotimes_{j \in I^\mathsf{c}} g_\eta)\).
An example of such a graph state is represented figure~\ref{fig:graph_state}.
We shall use the structure of the open graph to study computations using the
graph state. These graph states admit approximate stabilisers, which can be
understood as stabilisers in the infinite-squeezing limit:
\begin{Lem}[Approximate stabilisers]
  \label{lem:approx-stab}
  Let \((G,I,O)\) be an open graph, then, for any \(k \in I^\mathsf{c}\) and \(s
  \in \RR\),
  \begin{equation}
    \lim_{\eta \to \infty} \norm{ \Uch_{\Xo_j(s) \prod_{k \in N(k)} \Zo_k(A_{j,k} s)} \circ \Gch_\eta[\rho] - \Gch_\eta[\rho]} = 0.
  \end{equation}
\end{Lem}
\begin{proof}
  Since we do not actually need this lemma to prove our main results, we omit
  its proof, which is in any case a simpler version of the proof of
  lemma~\ref{lem:cont-stab}.
\end{proof}

The main results of this article correspond to direct generalisations of
proposition \ref{prop:tele-circuit} to measurement procedures over arbitrary
graph states.
The question is: given a graph state, is there an order to measure the vertices
(of the graph state) in such that we can always correct for the resulting
measurement errors?
In such a scheme, the measurement error spreads over several edges to more than
one adjacent vertex and we need a more subtle correction strategy, culminating in
our definition of CV-flow and a corresponding correction protocol.
In section \ref{sec:flow}, we exhibit our CV-flow condition and a corresponding 
MBQC procedure, in theorem \ref{thm:CV-flow}.
Then, in section \ref{sec:circuit} of the paper, we extract the unitary 
implemented by this MBQC protocol, proving a direct equivalent to proposition
\ref{prop:tele-circuit}: when \(\abs{I} = \abs{O}\), the protocol converges
strongly to the extracted unitary which acts on the input state.
This is the content of theorem \ref{thm:CV-circuit}.

\section{Correction procedures}
\label{sec:flow}
We are now ready to begin our study of CV graph states for CV-MBQC.
We first show that the original flow condition, causal flow
\cite{danos_determinism_2006}, also holds in continuous variables and results in
a nearly identical MBQC protocol.
We then state a generalised condition which we call CV-flow, inspired by g-flow
but valid for continuous variables and different to g-flow.
These conditions are associated with corresponding CV-MBQC correction protocols.
While the CV-flow protocol subsumes the original flow protocol as far as CV-MBQC
is concerned, causal flow is still worth understanding on its own in this
context, not least because the proof of theorem \ref{thm:CV-circuit} reduces the
CV-flow case to the causal flow case.

\paragraph{\sffamily Appendices.}
  In addition to the content of this section, appendix \ref{app:algorithm}
  describes a polynomial-time algorithm for determining a CV-flow for an open
  graph, whenever it has one, following almost exactly the qubit case by Mhalla
  and Perdrix \cite{mhalla_finding_2008}.
  Appendix \ref{app:comparison} contains a comparison between CV-flow and the
  original g-flow condition for qubits, when this makes sense.
  Appendix \ref{app:depth} contains an example of depth-complexity advantage
  using CV-MBQC compared to a circuit acting on the same number of inputs.

\subsection{Causal flow in continuous variables}
\label{ssec:simple}

In this section, we see how the causal flow condition extends to continuous
variables
We additionally allow for weighted graphs, but this does not change the
definition.

\begin{Def}[\cite{danos_determinism_2006}]
  An open graph \((G,I,O)\) has \textbf{causal flow} if there exists a map
  \(f:O^{\mathsf{c}} \to I^{\mathsf{c}}\) and a partial order \(\prec\) over
  \(G\) such that for all \(i \in O^{\mathsf{c}}\):
  \begin{itemize}
  \item \(A_{i,f(i)} \neq 0\);
  \item \(i \prec f(i)\);
  \item for every \(k \in N\big(f(i)\big) \setminus \{i\},\) we have \(i \prec
    k\).\qedhere
  \end{itemize}
  \vspace{-3mm}
  \label{def:flow}
\end{Def}
The order \(\prec\) is interpreted as a measurement order for the MBQC.
When an open graph has causal flow \((f, \prec)\), the function \(f\) identifies for each
measurement of a vertex \(j\) a single, \emph{as-of-yet unmeasured} neighbouring
vertex \(f(j)\) on which it is possible to correct for the measurement error.
This renders the corresponding CV-MBQC protocol: after constructing the
graph state,
\begin{enumerate}
\item measure the non-output vertices in the graph, in any order which is a linear
  extension of \(\prec\), in the basis corresponding to \(\Uo(\a, \b,
  \g)\)---for example, by applying \(\Uo(\a, \b, \g)\) and measuring \(\Po\);
  and,
\item immediately after each measurement (say of vertex \(j\)), and before any
  other measurement is performed, correct for the measurement error \(m_j\) onto
  \(f(j)\), by applying
  \begin{equation} \label{eqn: correction flow}
    \Co_j(m_j) \coloneqq \Xo_{f(j)}(- A_{j,f(j)}^{-1} m_j) \; \smashoperator{ \prod_{k \in N(f(j)) \setminus \{j\}}} \;
    \Zo_k(- A_{j,f(j)}^{-1} A_{f(j),k} m_j).
  \end{equation}
\end{enumerate}

We will now see how (\ref{eqn: correction flow}) can be viewed as an a-causal correction, through the completion of a stabiliser (as in the discrete case \cite{danos_determinism_2006,browne_generalized_2007,markham_entanglement_2014}). We first note that the state after measuring vertex $j$ and getting result $m_j$ is
equivalent to the state if one had first applied $\Zo_j(-m_j)$ and then measured, getting
result $0$. Since the result $0$ corresponds to the ideal computation branch,
the correction ideally corresponds to undoing this $\Zo_j(-m_j)$. However, this
cannot be done on vertex $j$ because we would have to do it before the
measurement result were known.
From this point of view the above correction acts to a-causally
implement $\Zo_j(-m_j)$ through the stabiliser relation.

To see this we note that the unitary
\begin{equation}
  \label{eq:causal_flow_corrections}
  \Xo_{f(j)}(- A_{j,f(j)}^{-1} m_j) \; \smashoperator{ \prod_{k \in N(f(j))}} \;
    \Zo_k(- A_{j,f(j)}^{-1} A_{f(j),k} m_j)= \Co_j(m_j) \Zo_{f(j)}(m_j)
\end{equation}
is a stabiliser for the graph state in the infinite squeezing limit---that is,
its action leaves the state unchanged, as described in lemma
\ref{lem:approx-stab}. Thus, application of $\Co_j(m_j)$ is, in the limit,
equivalent to applying $\Zo_{f(j)}(m_j)$.

In this way, we can see the third condition in the definition as merely imposing that all the vertices acted upon by \(\Co_j\) have not been measured when we try to complete this
stabiliser.

To see that the correction (\ref{eqn: correction flow}) affects the a-causal correction of $\Zo_j(-m_j)$ another way, we can commute the correction through the graph
operations:
\begin{align}
  \Co_j(m_j) \smashoperator{ \prod_{k \in N(f(j))} } \Co\Zo_{j,k}(A_{j,k})
  \;
  &= \Xo_{f(j)}(- A_{j,f(j)}^{-1} m_j) \; \smashoperator{ \prod_{k \in N(f(j)) \setminus \{j\}}} \;
    \Zo_k(- A_{j,f(j)}^{-1} A_{f(j),k} m_j) \cdot \smashoperator{ \prod_{k \in N(f(j))} } \Co\Zo_{j,k}(A_{j,k}) 
  \; \\
  &= \; \smashoperator{\prod_{k \in N(f(j)) \setminus \{j\}}} \Co\Zo_{f(j),k}(A_{j,k}) \;
  \Zo_j(m_j) \; \Co\Zo_{f(j),j}(A_{j,f(j)}).\label{eq:causal_flow_corrections_commutation}
\end{align}
In this picture, we see that the correction \(\Co_j(m_j)\) has a
\emph{back-action} \(\Zo_j(m_j)\) on the vertex \(j\) even though it has already
been measured.
This back-action appears before the measurement, even though the correction is
applied after the measurement.
%In this way, we can view \(\Co_j(m_j)\) as an acausal correction which corrects for the measurement after it has been made: when commuted through the graph5operations it maps the outcome \(m_j\) of the measurement to the ``ideal'' outcome \(m_j = 0\) since \(\Zo_j(m_j) \Zo(-m_j) = \mathsf{I}\) and \(\Zo(m_j) \Po_j \Zo(-m_j) = \Po_j - m_j\).

Let \(\vec{\a},\vec{\b},\vec{\g} \in \RR^{\abs{O^\mathsf{c}}}\) identify
measurement angles for each non-output mode as in step 1, then, for a given open
graph with causal flow, we denote \(\Fch_\eta(\vec{\a},\vec{\b},\vec{\g})\)
the quantum map corresponding to this MBQC procedure starting with the
corresponding graph state with local squeezing factor \(\eta\). Using the same
trick as for the gate teleportation, for any input state \(\rho \in
D(\mathcal{H}^{\otimes \abs{I}})\) we can write this quantum channel as:
\begin{equation}
  \label{eq:causal_flow_channel}
  \begin{aligned}
    &\Fch_\eta(\vec{\a},\vec{\b},\vec{\g})[\rho]\\
    &= \left( \prod_{j \in
        O^\mathsf{c}}^{\prec} \tr_j \circ \Uch_{\exp(-i\Po_j \cdot \Po_{f(j)})} \circ
      \Uch_{\exp(- i\Po_v \sum_{k \in N(f(j)) \setminus \{j\}} A_{j,f(j)}^{-1} A_{f(j),k} \Qo_k)} \circ \Uch_{\Uo_v(\a_v,\b_v,\g_v)} \right) \\
    & \quad \quad \quad \quad \quad \circ \Gch_\eta[\rho],
  \end{aligned}
\end{equation}
where the product is to be interpreted as sequential composition of channels,
and is ordered by the measurement order \(\prec\). Here once again, the trace
over the subspace of qumode \(j\) corresponds to the measurement of that qumode,
and the 2-qumode unitaries simulate the classically-controlled corrections in
the protocol and described by equation~\eqref{eq:causal_flow_corrections}.
Namely, \(\exp(-i\Po_j \cdot \Po_{f(j)})\) simulates the \(\Xo\) part of the
correction, and 
\begin{equation}
  \exp(- i\Po_v \sum_{k \in N(f(j)) \setminus \{j\}} A_{j,f(j)}^{-1} A_{f(j),k} \Qo_k)
\end{equation}
simulates all of the \(\Zo\) parts on neighbours, which commute thus can be
written as a single exponential.

Then, we have the following:
\begin{Prop}[Causal flow protocol]
  Suppose the open graph \((G,I,O)\) has causal flow, then for any \(\vec{\a},
  \vec{\b}, \vec{\g} \in \RR^{\abs{O^\mathsf{c}}}\) and any input state, the
  corresponding MBQC procedure \(\Fch_\eta(\vec{\a},\vec{\b},\vec{\g})\) is
  runnable: no corrections depend on the outcome of measurements before they are
  made, and no corrections are made on vertices after they are measured.
  \label{prop:flow}
\end{Prop}
\begin{proof}
  This is clear by the condition that \(i \prec f(i)\): we always measure node
  \(i\) before the vertex \(f(i)\) onto which we perform the corresponding
  correction.
\end{proof}

\subsection{CV-flow}
\label{ssec:cv-flow}

% \begin{figure}
%   \centering
%   \begin{center}
%     \scalebox{1.1}{\tikzfig{figures/examples/cv-flow-ex}}
%   \end{center}
%   \caption{Example of a correction procedure based on a CV flow. We perform
%     measurements on the black vertices of a graph state with measurement outcomes
%     \(m_1,m_2,m_3\), the white vertices are unmeasured (top left). The linear
%     equations for the correction matrix of the graph (top right) is solvable for
%   all measurement ouctomes (bottom left), which gives a corresponding correction
%   procedure (bottom right).}
%   \label{fig:CV_flow}
% \end{figure}
\begin{figure}
  \centering
  \begin{center}
    \scalebox{.90}{\input{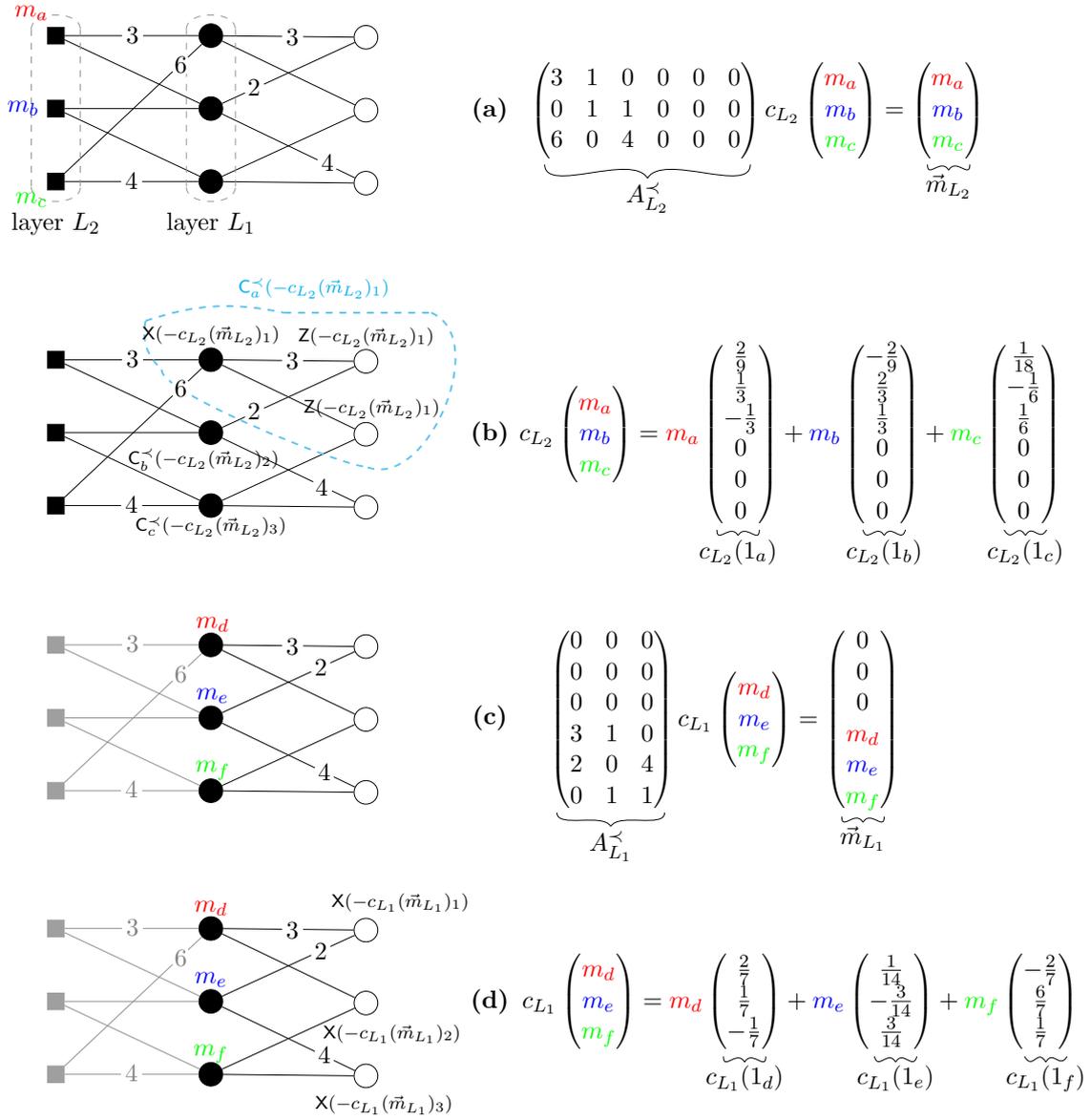}}
  \end{center}
  \caption{Example of the MBQC procedure based on a CV-flow.
    \textbf{(a)} We start with a graph state with a candidate measurement order: the
    vertices in layer 2 must be measured before layer 1.
    As explained in section \ref{sec:circuit}, all CV-flows can be broken up
    into such a sequence of layers such that all the vertices within a layer can be
    measured simultaneously.
    We perform measurements on the vertices in layer 2 (in a given basis for the
    unitary \eqref{eq:unitary}),
    obtaining measurement outcomes \(\vec{m}_{L_2} = (m_a,m_b,m_c)^T\) and
    leading to the correction equation on the right of (a), which is a direction
    application of equation \eqref{eq:layer-correction}.
    \textbf{(b)} This linear equation has a solution for any measurement
    outcomes giving a CV-flow for layer 2 as per definition \ref{def:CV-flow},
    which we decompose into the different contributions from the measurement of
    each vertex as per lemma \ref{lem:cv-flow-multilinear} (right).
    This leads to a correction procedure for the measurements in the MBQC
    procedure (left).
    The corrections \(\Co\) take the form of partial stabilisers as described in
    equations \eqref{eq:correction} and \eqref{eq:layer_correction}.
    \textbf{(c)} We then measure the vertices in the layer 1, with measurement outcomes
    \(\vec{m}_{L_1} = (m_d,m_e,m_f)^T\).
    The previously measured vertices, which are no longer accessible for corrections,
    have been grayed out.
    This second set of measurements has its own correction equation.
    \textbf{(d)} The solution to the linear equation (see \eqref{eq: layer decomp correction} and \eqref{eq:layer-correction}), and corresponding correction
    procedure for layer 1.
    Since at this point, we have measured all the vertices in the graph, the MBQC
    procedure is complete.
    Furthermore, as we have seen, for this measurement order it is possible to
    correct for any measurement error at any of the measurements and the
    open graph has a CV-flow.}
  \label{fig:CV_flow}
\end{figure}

In \textcite{browne_generalized_2007} a more general notion of flow was presented for qubit MBQC, by loosening the conditions of causal flow.
In particular, this is done by allowing corrections to be applied on more than one neighbour
and loosening the condition on neighbours of the correction vertex - they no longer need to all be unmeasured.
This is possible by selecting vertices on which to correct in such a way that their
different contributions add up to correct the measurement error.
As well as allowing for more general correction strategies, which can, for example decrease the depth of a computation (number of measurement rounds), it also allows for measurements on different Pauli planes (the X-Y, X-Z, Y-Z planes).

This leads to the definition of \textit{g-flow}, which is similar to Definition \ref{def:flow} above for causal flow, but where the single vertex $f(i)$ is replaced by a set of vertices, and the conditions on this set are a bit more involved, and depend on the choice of basis measurement (see \cite{browne_generalized_2007} and Definition \ref{def:gflow} in section \ref{qudit_flow:flow:gflow}).
It can simply be understood from the stabiliser perspective, that this set of vertices correspond to the set of stabilisers that should be applied to correct the measurement error. Then the conditions on their choice are such that i) they can do this correction (i.e. depending on the plane, the corrections are equivalent to applying a Z, Y, or X Pauli on the measured qubit, via the stabiler relation), and that ii) they should not interfere with previous corrections (i.e. the stabiliser has no part on measured qubits).

In the extension to continuous variables, we instead phrase our conditions in terms of matrices. 
This form can be used for qubits also, but in the case of CV (and qudits \cite{booth_outcome_2021}) it is much simpler to understand than writing a form similar to \ref{def:flow}. For this reason, we now develop the matrix formulation, rather than present the original form of \cite{browne_generalized_2007}. In section \ref{qudit_flow:flow:gflow} the form of \cite{browne_generalized_2007} can be found, along with the reduction of the matrix form to the original form for qubits.

To begin, given a candidate measurement order, we consider, at each step of that order, a partition, or cut, of the graph into 2 subsets: vertices that have yet to be measured (that can be used for corrections),
and vertices that have already been measured (for which we need to be careful
about unwanted back-actions).

If \((X,Y)\) is a pair of subsets of \(V\), we define \(A[X,Y]\) as the submatrix
obtained from the adjacency matrix of \(G\) by keeping only the rows 
corresponding to vertices in \(X\) and the columns corresponding to vertices in 
\(Y\).

\begin{Def}
  Let \((G,I,O)\) be an open graph, \(<\) a total order over
  the vertices \(V\) of \(G\) and \(A^<\) be the adjacency matrix of \(G\) such that its
  columns and rows are ordered by \(<\).
  Further, define \(P(j)\) (the ``past'' of node \(j\)) as the subset of \(V\)
  such that \(k \in P(j)\) implies \(k \leqslant j\).
  Then the \textbf{correction matrix} \(A_j^<\) of a vertex \(j \in V\) is the 
  matrix \(A_j^< \coloneqq A[P(j), (P(j) \cup I)^\mathsf{c}]\).
  \label{def:matrix}
\end{Def}

The correction matrix tells us how \(\Xo\) and \(\Zo\) operations on unmeasured
nodes are going to affect the previously measured nodes, thus it allows us to
determine how to apply a correction on a specific vertex by controlling this
back-action.

\begin{Def}
  An open graph \((G,I,O)\) has \textbf{CV-flow} if there exists a
  partial order \(\prec\) on \(O^{\mathsf{c}}\) such that for any total order
  \(<\) that is a linear extension of \(\prec\) and every \(j
  \in O^\mathsf{c}\), there is a function \(c_j : \RR \to  \RR^{\abs{V} -
    \abs{P(j) \cup I}}\) such that for all \(m \in \RR\) the linear equation
  \begin{equation}
    A_j^< c_j(m) =
    \begin{pmatrix}
      0 \\ \vdots \\ 0 \\ m
    \end{pmatrix}
    \qq{holds,}\
    \label{eq:correction_eq}
  \end{equation}
  where \(A_j^<\) is the correction matrix of vertex \(j\).
  Letting \((c_j)_{j \in O^\mathsf{c}}\) be such a set of functions, we
  call the pair \((\prec, (c_j))\) a CV-flow for \((G,I,O)\).
  \label{def:CV-flow}
\end{Def}

This definition is somewhat convoluted in order to allow the formalism to
describe partial orders of measurement rather than only total orders (which
would give a simpler definition). This is so that we can take into account
situations where several qumodes can be measured simultaneously without
contradictions (see lemma~\ref{lem:cv-flow-multilinear} and
appendix~\ref{app:depth}).

\begin{Lem}
  If \((\prec, (c_j)_{j \in O^\mathsf{c}})\) is a CV-flow for an open graph
  \((G,I,O)\), then the functions \(c_j\) can be chosen as \(\RR\)-linear, that
  is, for each  \(i \in  O^\mathsf{c}\) and any \(m \in \RR\), \(c_j(m) = m
  \cdot c_j(1)\).
  \label{lem:cv-flow-linear}
\end{Lem}

\begin{proof}
  Let \(c_j(1)\) be such that equation \eqref{eq:correction_eq} holds.
  Then,
  \begin{equation}
    A_j^< (m \cdot c_j(1)) = m \cdot A_j^< c_j(1) =
    m \cdot \begin{pmatrix}
      0 \\ \vdots \\ 0 \\ 1
    \end{pmatrix} = 
    \begin{pmatrix}
      0 \\ \vdots \\ 0 \\ m
    \end{pmatrix}\qedhere
  \end{equation}
\end{proof}

Now, there is a correction procedure on an open graph \((G,I,O)\) if there is
a measurement order such that each measurement error can be corrected for
without any backaction on previously measured vertices and such an order is
guaranteed by the existence of a CV-flow.
For any open graph with CV-flow, the MBQC protocol follows from generalising the
causal flow protocol:
\begin{enumerate}
\item measure the non-output vertices in the graph in any order which is a linear
  extension of \(\prec\), in the basis corresponding to \(\Uo(\a, \b, \g)\);
  and,
\item immediately after each measurement (say of vertex \(j\)), and before any
  other measurement is performed, correct for the measurement error \(m_j\) by
  applying
  \begin{equation}
    \Co_j^\prec(m_j) \coloneqq \prod_{k \in V \setminus P(j) \cup I}
    \left( \Xo_k(-c_j(m_j)_k)
      \quad\smashoperator{\prod_{\ell \in N(k) \setminus P(j)}}\quad \Zo_\ell(-A_{k,\ell} 
      \cdot c_j(m_j)_k) \right),
    \label{eq:correction}
  \end{equation}
  where \(c_j(m_j)_k\) is the \(k\)-th element of the vector \(c_j(m_j) \in 
  \RR^{\abs{V} - \abs{P(j) \cup I}}\).
\end{enumerate}

Keen readers will note that performing the corrections following
equation~\eqref{eq:correction} might require performing very large number of
Pauli operations (performing all corrections in the MBQC would be approximately
\(O(\abs{V}^3)\) Pauli operations). This form gives a better intuition for what
is going in terms of partial graph stabilisers, but it can be simplified as
follows. The idea is that partial stabilisers for the same measurement error
commute: pulling the partial stabilisers through the graph operations to meet the
measurement error (as in
equation~\eqref{eq:causal_flow_corrections_commutation}), one recovers full
graph stabilisers, which are known to commute. Then, we can sum all of the
\(\Zo\) corrections that act on the same vertex, resulting in the reduced form
for equation~\eqref{eq:correction}:
\begin{equation}
  \label{eq:correction_simplified}
  \Co_j^\prec(m_j)
  \coloneqq
  \left( \prod_{k \in V \setminus P(j) \cup I} \Xo_k(-c_j(m_j)_k) \right)
  \left( \prod_{k \in V \setminus P(j) \cup I} \Zo_\ell(-\sum_{\ell \in V} A_{k,\ell} \cdot c_j(m_j)_k) \right),
\end{equation}
which requires only \(O(\abs{V}^2)\) corrections for the MBQC. Typically
however, the number of corrections will be much smaller than this since we only
ever act on the neighbours of neighbours of the measured vertex at each step,
and graphs which are too connected are unlikely to have CV-flow since then the
equations~\eqref{eq:correction_eq} become inconsistent.

In fact, under this second description of corrections, one can show that the
number of corrections is upper-bounded by \(\abs{V}(\abs{V}-1)\). To prove this
would involve introducing more technicalities which are beyond the scope of this
article. Instead, we refer to \cite{booth_outcome_2021} which treats the
question of qudit MBQC but whose arguments lift immediately to CV if one uses
the same analogy as we do in section~\ref{sec:qudit} (but in the opposite
direction).

\begin{Rem}
  It is straightforward how to extend these definitions to simultaneously
  correct for a subset \(L \subseteq O^\mathsf{c}\) of vertices \emph{unrelated by
    \(\prec\)}: for any total linear extension \(<\) of \(\prec\) let
  \(P(L)\) be the subset of \(V\) such that whenever \(k \in P(L)\), for some
  \(j \in L\) we have \(k \leqslant j\).
  We can then define the correction matrix \(A_L^< = A^<[P(L),V \setminus
  (P(L) \cup I)]\) of \(L\) identically to definition \ref{def:matrix}.
  There is a simultaneous correction procedure for \(L\) if  there is a function
  \(c_L : \RR^\abs{L} \to \RR^{\abs{V} - \abs{P(L)\cup I}}\) such that
  \begin{equation}
    A_L^< c_L(\vec{m}) =
    \begin{pmatrix}
      0 \\ \vdots \\ 0 \\ \vec{m}
    \end{pmatrix}
    \qq{for all potential measurement outcomes} \vec{m} \in \RR^\abs{L}.
    \label{eq:layer-correction}
  \end{equation}

  From the proof of lemma \ref{lem:cv-flow-linear}, we can also see that
  \begin{Lem}[CV-flow linearity]
    The function \(c_L\) from equation \eqref{eq:layer-correction} can be chosen as 
    \(\RR\)-multilinear, in the following sense: for any \(\vec{m} \in 
    \RR^\abs{L}\),
    \begin{equation} \label{eq: layer decomp correction}
      c_L(\vec{m}) = \sum_{k \in L} m_k \cdot c_L(1_k),
    \end{equation}
    where \(1_k\) is the column vector with a single \(1\) in its  \(k\)-th row
    and \(0\) elsewhere.
    \label{lem:cv-flow-multilinear}
  \end{Lem}
  
  It is therefore possible to measure several vertices at once in step 1 of the 
  procedure, so long as one never measures two vertices comparable by \(\prec\) in 
  a single measurement step.
  The correction then takes the form of a product of corrections of the form of
  equation \eqref{eq:correction}:
  \begin{equation}
    \Co_L^\prec(\vec{m}) \coloneqq \prod_{j \in L} \Co_j^\prec(m_j),
    \label{eq:layer_correction}
  \end{equation}
  and where \(C_j^\prec(m_j)\) is calculated using \(c_j(m_j) = m_j \cdot
  c_L(1_j)\) by lemma \ref{lem:cv-flow-multilinear}.
  We will use this form in section \ref{sec:circuit}.
  A simple example is worked out in figure \ref{fig:CV_flow}.
\end{Rem}

As for causal flow, let \(\vec{\a},\vec{\b},\vec{\g} \in
\RR^{\abs{O^\mathsf{c}}}\) identify measurement angles for each non-output mode
as in step 1, then, for a given open graph with causal flow, 
\(\Cch_\eta^\d(\vec{\a},\vec{\b},\vec{\g})\) denotes the quantum map
corresponding to this MBQC procedure starting with the corresponding graph state
with local squeezing \(\eta\). Once again we can write this channel
explicitly by replacing the corrections with multi-qumode unitary gates
simulating classical control:
\begin{equation}
  \begin{aligned}
    &\Cch_\eta(\vec{\a},\vec{\b},\vec{\g})[\rho]\\
    &= \left( \prod_{j \in
        O^\mathsf{c}}^{\prec} \tr_j \circ \Uch_{\exp(-i\Po_j c_j(1) \cdot \vec{\Po})} \circ
      \Uch_{\exp(- i\Po_v \sum_{k \in N(f(j)) \setminus \{j\}} A_{j,f(j)}^{-1} A_{f(j),k} \Qo_k)} \circ \Uch_{\Uo_v(\a_v,\b_v,\g_v)} \right) \\
    & \quad \quad \quad \quad \quad \circ \Gch_\eta[\rho],
  \end{aligned}
\end{equation}
where
\begin{equation}
  c_j(1) \cdot \vec{\Po} = \sum_{k \in V \setminus (P(j) \cup I)} c_j(1)_k \Po_k.
\end{equation}

\begin{Thm}[CV-flow protocol]
  \label{thm:CV-flow}
  Suppose the open graph \((G,I,O)\) has CV-flow, then for any \(\vec{\a},
  \vec{\b}, \vec{\g} \in \RR^{\abs{O^\mathsf{c}}}\) and any input state, the
  corresponding MBQC procedure \(\Cch_\eta(\vec{\a},\vec{\b},\vec{\g})\) is
  runnable.
\end{Thm}

\begin{proof}
  This is true by construction, since we only ever consider corrections for a
  given measurement error that act on unmeasured nodes at that step in the MBQC
  procedure.
\end{proof}
% The idea is essentially the same as proposition \ref{prop:flow}, but with more
% more gates to consider: by commuting the corrections through the \(\Co\Zo\)
% gates, one can see that the protocol is equivalent to a sequence of parallel gate
% teleportations.

\subsection{Recovering g-flow}
\label{qudit_flow:flow:gflow}

Generalised flow, or \emph{g-flow}, is the flow condition which inspired
CV-flow. It was originally formulated in terms of parity conditions on the
connectivity of the open graph for the case of MBQC with qubits
\cite{browne_generalized_2007}. The Hilbert space of the qubit can be viewed as
a space of functions \(\ZZ_2 \to \CC\), and thus the open graphs in question are
\(\ZZ_2\)-edge-weighted. However, note that an open \(\ZZ_2\)-graph  can
equivalently be viewed as an unweighted open graph. If \(A \subseteq V\), we
write \(\operatorname{Odd}(A)\) the subset of \(\bigcup_{a \in A} N(a)\) of
vertices that are neighbours of an odd number of elements of \(A\), then:
\begin{Def}[\cite{browne_generalized_2007}]
  An open \(\ZZ_2\)-graph \((G,I,O)\) has \emph{g-flow} if there exists a map
  \(g : O^\mathsf{c} \to 2^{I^\mathsf{c}}\) and a partial order \(\prec\) on
  \(V\) such that for all \(i \in O^\mathsf{c}\),
  \begin{itemize}
  \item if \(j \in g(i)\) and \(i \neq j\) then \(i \prec j\);
  \item if \(j \preceq i\) and \(i \neq j\) then \(j \notin
    \operatorname{Odd}(g(i))\);
  \item \(i \notin g(i)\) and \(i \in \operatorname{Odd}(g(i))\). \qedhere
  \end{itemize}
\end{Def}
The idea is that, in the case of qubits, the Pauli operations used for
correction are all self-inverse, \(\Xo^2 = \Io\) and \(\Zo^2 = \Io\), and the
usual graph stabilisers are used for corrections: \(\Ko_j = \Xo_j \prod_{k \in
  N(j)} \Zo_k\) for any \(j \in I^\mathsf{c}\). As a result, \(\Ko_j^2 = \Io\).
g-flow has the same interpretation as CV-flow: it controls the
back-action of partial stabilisers acting on unmeasured vertices at each step of
the measurement procedure.

Making this interpretation explicit, these parity conditions can
straightforwardly be reinterpreted as linear equations over \(\ZZ_2\), involving
(submatrices of) the adjacency matrix of the graph. Using the same notation as
for CV-flow:
\begin{Prop}
  An open graph \((G,I,O)\) has g-flow \((g,\prec)\) if and only if for any
  total order \(<\) that is a linear extension of \(\prec\) and every \(j
  \in O^\mathsf{c}\), there is a function \(c_j : \ZZ_2 \to  \ZZ_2^{\abs{V} -
    \abs{P(j) \cup I}}\) such that for all \(m \in \ZZ_2\) the linear equation
  \begin{equation}
    A_j^< c_j(m) =
    \begin{pmatrix}
      0 \\ \vdots \\ 0 \\ m
    \end{pmatrix}
    \qq{holds,}\
    \label{eq:correction_eq}
  \end{equation}
  where \(A_j^<\) is the correction matrix of vertex \(j\).
\end{Prop}
Hopefully, it should be clear that CV-flow, definition~\ref{def:CV-flow},
corresponds exactly to this formulation of g-flow, where one replaces the
finite field \(\ZZ_2\) with \(\RR\).  

The proof of this statement follows from identifying each set \(g(j)\) with a
column vector \(v^j \in \ZZ_2^{\abs{V} - \abs{P(j) \cup I}}\):
\begin{equation}
  v^j_u =
  \begin{cases}
    1 \qif u \in g(j); \\
    0 \quad \text{otherwise}.
  \end{cases}
\end{equation}
Then setting \(c_j(m) = m \cdot v^j\) gives the desired function. We omit the
formal proof as it mainly involves technical calculation of matrix elements
which are therefore not particularly informative.

\section{Circuit extraction}
\label{sec:circuit}
We have shown that a correction procedure is possible when the open graph has
CV-flow.
We now address the induced quantum map and the question of convergence of these
teleportations.
In general, for an arbitrary graph state, even one with CV-flow, the MBQC
procedure is not convergent.
It is possible for the output state to contain squeezing dependant components
which diverge in the limit.

We will see that for an open graph \((G,I,O)\) with \(\abs{I} < \abs{O}\),
an MBQC protocol is equivalent to one based on an open graph \((G, I', O)\), where
\(I \subseteq I'\) and \(\abs{I'} = \abs{O}\).
The additional states in \(I'\) are auxiliary squeezed states \(g_\eta\)
reinterpreted as inputs, and this forces divergence from the inputs to the
outputs. The simplest example is given by the single-vertex open graph
\(\begin{tikzpicture}
	\begin{pgfonlayer}{nodelayer}
		\node [style={white_circle}] (0) at (0, 0) {};
	\end{pgfonlayer}
\end{tikzpicture}
\) with no input and a single
output.\footnote{This open graph is formally defined on the vertex set
  \(\{\begin{tikzpicture}
	\begin{pgfonlayer}{nodelayer}
		\node [style={white_circle}] (0) at (0, 0) {};
	\end{pgfonlayer}
\end{tikzpicture}
\}\) as
  \((0,\varnothing,\{\begin{tikzpicture}
	\begin{pgfonlayer}{nodelayer}
		\node [style={white_circle}] (0) at (0, 0) {};
	\end{pgfonlayer}
\end{tikzpicture}
\})\). That is, the
  adjacency matrix of the graph is \(G=0\) (since there can be no self-edges from
  the only vertex to itself), there are no inputs \(I=\varnothing\), and the only
  vertex is an ouput, \(O=\{\begin{tikzpicture}
	\begin{pgfonlayer}{nodelayer}
		\node [style={white_circle}] (0) at (0, 0) {};
	\end{pgfonlayer}
\end{tikzpicture}
\}\)} One
can readily check that it trivially has CV-flow but the output is a
squeezed state. As a result, when if we try to take the infinite squeezing limit
the state diverges in \(\mathscr{H}\). 

As we shall see, the existence of a CV-flow also implies both that \(\abs{I} 
\leqslant \abs{O}\) and \(\abs{O} \neq 0\), hence we conclude that a necessary
and sufficient condition for the protocol to be convergent is that the open 
graph has as many input vertices as outputs, \(\abs{I} = \abs{O} \neq 0\). This
clearly excludes the simple counter-example presented above.

The proof will proceed by explicitely constructing the limit of the protocol,
via a circuit extraction scheme inspired by \cite{miyazaki_analysis_2015}.
While the circuits extracted by our scheme as well as the broad structure of the
proof are entirely analogous to that work, our proof method is quite different.
The original graph-theoretical arguments using local complementation in
\cite{miyazaki_analysis_2015} are delicate to adapt to CV, so we reason instead
with the adjacency matrix and correction matrices of the open graph. This method
extends the DV case, and should also apply to more general MBQC scenarios beyond
CV. In particular, it easily applies to qudit MBQC for any prime local
dimension, as we shall see in section~\ref{sec:qudit}.

\subsection{Star pattern transformation}
\label{sec:spt}

\begin{figure}
  \centering
  \begin{center}
    \scalebox{.75}{\begin{tikzpicture}
	\begin{pgfonlayer}{nodelayer}
		\node [style={black_square}] (0) at (4.25, 5.8) {};
		\node [style={black_square}] (1) at (4.25, 3.55) {};
		\node [style={black_square}] (2) at (4.25, 1.3) {};
		\node [style={black_circle}] (3) at (7.25, 5.8) {};
		\node [style={black_circle}] (4) at (7.25, 3.55) {};
		\node [style={white_circle}] (5) at (7.25, 1.3) {};
		\node [style={white_circle}] (6) at (10.25, 3.55) {};
		\node [style={white_circle}] (7) at (10.25, 5.8) {};
		\node [style={black_square}] (10) at (13.5, 5.8) {};
		\node [style={black_square}] (11) at (13.5, 3.55) {};
		\node [style={black_square}] (12) at (13.5, 1.3) {};
		\node [style={black_circle}] (13) at (16.75, 5.8) {};
		\node [style={black_circle}] (14) at (16.75, 3.55) {};
		\node [style={white_circle}] (15) at (16.75, 1.3) {};
		\node [style={white_circle}] (16) at (20, 3.55) {};
		\node [style={white_circle}] (17) at (20, 5.8) {};
		\node [style=none] (20) at (23.25, 5.75) {};
		\node [style=none] (21) at (31.25, 5.75) {};
		\node [style=none] (22) at (23.25, 3.5) {};
		\node [style=none] (23) at (31.25, 3.5) {};
		\node [style=none] (24) at (23.25, 1.25) {};
		\node [style=none] (25) at (31.25, 1.25) {};
		\node [style=ctrl] (26) at (27.25, 5.75) {};
		\node [style=ctrl] (27) at (27.25, 3.5) {};
		\node [style=ctrl] (28) at (24.75, 3.5) {};
		\node [style=ctrl] (29) at (24.75, 1.25) {};
		\node [style=ctrl] (30) at (29.75, 3.5) {};
		\node [style=ctrl] (31) at (29.75, 5.75) {};
		\node [style=none] (34) at (33.25, 5.75) {};
		\node [style=none] (35) at (43.25, 5.75) {};
		\node [style=none] (36) at (33.25, 3.5) {};
		\node [style=none] (37) at (43.25, 3.5) {};
		\node [style=none] (38) at (33.25, 1.25) {};
		\node [style=none] (39) at (43.25, 1.25) {};
		\node [style=ctrl] (40) at (38.25, 5.75) {};
		\node [style=ctrl] (41) at (38.25, 3.5) {};
		\node [style=ctrl] (42) at (35.75, 3.5) {};
		\node [style=ctrl] (43) at (35.75, 1.25) {};
		\node [style=ctrl] (44) at (40.75, 3.5) {};
		\node [style=ctrl] (45) at (40.75, 5.75) {};
		\node [style=gate] (46) at (39.5, 5.75) {\(\Jo_{12}\)};
		\node [style=gate] (47) at (37, 3.5) {\(\Jo_{45}\)};
		\node [style=gate] (48) at (39.5, 3.5) {\(\Jo_{56}\)};
		\node [style=gate] (49) at (34.5, 1.25) {\(\Jo_{78}\)};
		\node [style=gate] (50) at (42, 5.75) {\(\Jo_{23}\)};
		\node [style=none] (51) at (4.25, 6.625) {1};
		\node [style=none] (52) at (7.25, 6.625) {2};
		\node [style=none] (53) at (10.25, 6.625) {3};
		\node [style=none] (54) at (4, 4.25) {4};
		\node [style=none] (55) at (7.25, 4.25) {5};
		\node [style=none] (56) at (10.25, 4.25) {6};
		\node [style=none] (57) at (4.25, 2) {7};
		\node [style=none] (58) at (7.25, 2) {8};
		\node [style=none] (59) at (7.25, 0) {\textbf{(a)}};
		\node [style=none] (60) at (16.75, 0) {\textbf{(b)}};
		\node [style=none] (61) at (27.5, 0) {\textbf{(c)}};
		\node [style=none] (62) at (38.5, 0) {\textbf{(d)}};
	\end{pgfonlayer}
	\begin{pgfonlayer}{edgelayer}
		\draw (3) to (7);
		\draw [in=180, out=0] (0) to (3);
		\draw (0) to (4);
		\draw (4) to (1);
		\draw (5) to (1);
		\draw (2) to (5);
		\draw (6) to (3);
		\draw (4) to (6);
		\draw [style={black_edge}, in=180, out=0] (10) to (13);
		\draw (10) to (14);
		\draw (15) to (11);
		\draw (16) to (13);
		\draw [style={black_edge}] (13) to (17);
		\draw [style={black_edge}] (11) to (14);
		\draw [style={black_edge}] (14) to (16);
		\draw [style={black_edge}] (12) to (15);
		\draw (20.center) to (21.center);
		\draw (22.center) to (23.center);
		\draw (24.center) to (25.center);
		\draw (26) to (27);
		\draw (28) to (29);
		\draw (31) to (30);
		\draw (34.center) to (35.center);
		\draw (36.center) to (37.center);
		\draw (38.center) to (39.center);
		\draw (40) to (41);
		\draw (42) to (43);
		\draw (45) to (44);
	\end{pgfonlayer}
\end{tikzpicture}}
  \end{center}
  \caption{An example of circuit extraction of an MBQC for an open graph with
    causal flow, using star pattern transformation as described in section
    \ref{sec:spt}.
    Starting from an open graph with a causal flow \textbf{(a)}, we identify a path cover
    of the graph that agrees with the causal flow (thick edges directed
    \(u \to f(u)\) for each \(u \in O^\mathsf{c}\)) \textbf{(b)}, and each path as a wire
    in a quantum circuit \textbf{(c)}.
    The remaining edges of the graph implement \(\CZo\) gates in the circuit.
    Finally, each causal flow edge implements a unitary of the form given by
    equation \eqref{eq:j-gate}, and we obtain a circuit representation of the
    unitary implemented by the MBQC \textbf{(d)}.}
  \label{fig:spt}
\end{figure}
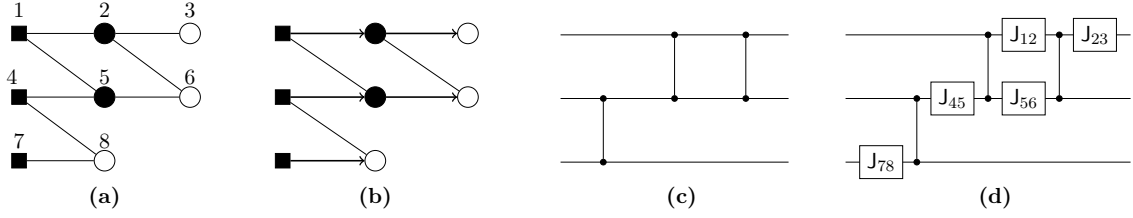

In order to model the computation through the MBQC, the trick is to distinguish
between ``real'' qumodes that undergo a unitary transformation though the MBQC
(which act like the wires in a circuit undergoing gates), and auxiliary qumodes
that are consumed in teleportations. In the case of causal flow, things work
quite nicely as follows.

We use the following which also holds in CV (since the causal flow does not depend on edge weights,
only the correction procedure):
\begin{Def}[\cite{de_beaudrap_finding_2008}]
  A \textbf{path cover} of an open graph \((G,I,O)\) is a collection
  \(\mathcal{P}\) of directed edges (or arcs) in \(G\) such that
  \begin{itemize}
  \item each vertex in \(G\) is contained in exactly one path in
    \(\mathcal{P}\);
  \item each path in \(\mathcal{P}\) is either disjoint from \(I\) or intersects
    \(I\) only at its initial point;
  \item  each path in \(\mathcal{P}\) intersects \(O\) only at its final point.
    \qedhere
  \end{itemize}
\end{Def}

\begin{Lem}[Causal flow path cover \cite{de_beaudrap_finding_2008}]
  Let \((f, \leqslant)\) be a causal flow on an open graph \((G,I,O)\).
  Then there is a path cover \(P_f\) of \((G,I,O)\) where \(x \to y\) is an
  arc in some path of \(P_f\) if and only if \(y = f(x)\).
  \label{lem:simple-cover}
\end{Lem}

Lemma \ref{lem:simple-cover} allows us to interpret the causal flow MBQC
procedure as a sequence of single qumode gate teleportations, with additional
entangling operations between teleportations.
In fact, the path cover \(\mathcal{P}_f\) allows us to distinguish between two
types of edges in \(G\):
\begin{itemize}
\item edges \((j,k) \in \mathcal{P}_f\) correspond to gate teleportations where
  one end is the input and the other the output;
\item edges \((j,k) \notin \mathcal{P}_f\) correspond to
  \(\CZo\big(A_{j,k}\big)\) gates in the final circuit.
\end{itemize}
\emph{Star pattern transformation} (STP) \cite{broadbent_parallelizing_2009} is a
method based on this intuition for turning the MBQC protocol on an open graph
with causal flow into an equivalent quantum circuit.
While it was originally formulated for DV, an almost identical method functions
in CV, the only real difference being the nature of the unitary gates.
Assume \((G,I,O)\) is an open graph with causal flow \((f,\prec)\) and 
corresponding path cover \(P_f\), and let
\begin{equation}
  \Jo(w,\a,\b,\g) \coloneqq \So(w) \Fo
  \Uo(\a,\b,\g), \qq{and for any subset \(S \subseteq V\),}
  \CZo_{j, S}(s) \coloneqq \prod_{k \in S} \CZo_{j, k}(w),
  \label{eq:j-gate}
\end{equation}
To obtain a circuit for the causal flow MBQC,
\begin{enumerate}
\item Interpret each path in \(P_f\) as a wire (qumode) in a quantum circuit,
  and index the wire by the collection of vertices intersected by the path.
\item For each edge \((j,k) \notin \mathcal{P}_f\), insert a
  \(\CZo\big(A_{j,k}\big)\) gate between the edges indexed by \(j\) and \(k\).
\item For each edge \((j,k) \in \mathcal{P}_f\), insert a
  \(\Jo\big(A_{j,k},\a,\b,\g\big)\) gate after all the \(\CZo\) gates for vertices \(i
  \in \mathcal{P}_f\) such that \(i \leqslant j\) but before all such gates for
  \(k \leqslant i\).
\end{enumerate}
An example is worked out in figure \ref{fig:spt}.
In the ideal limit, the MBQC procedure converges to the CV SPT map.
We have that:
\begin{restatable}[Causal flow circuit]{Prop}{FlowCircuit}
  Suppose the open graph \((G,I,O)\) has a causal flow and \(\abs{I} =
  \abs{O}\).
  Then for any \(\vec{\a}, \vec{\b}, \vec{\g} \in \RR^{\abs{O^\mathsf{c}}}\) and
  any \(\rho \in D(\Hcal^{\otimes \abs{I}})\),
  \begin{equation}
    \lim_{\eta \to \infty}
    \Fch_\eta(\vec{\a}, \vec{\b}, \vec{\g})[\rho]
    = \Uo_{SPT}(\vec{\a}, \vec{\b}, \vec{\g}) \rho \Uo_{SPT}^*(\vec{\a}, \vec{\b}, \vec{\g}),
  \end{equation}
  where \(\Uo_{SPT}\) is the unitary corresponding to the circuit obtained by
  star pattern transformation of \((G,I,O)\).
  Furthermore, the condition \(\abs{I} = \abs{O}\) is necessary.
  \label{prop:flow_convergence}
\end{restatable}
\begin{proof}
  The proof is left to appendix \ref{sapp:flow_convergence}.
\end{proof}

\subsection{CV-flow triangularisation}
\label{sec:cv-linear}

The next challenge is to do the same as above and extract a circuit, for CV-flow. From there one can treat convergence through viewing it as a sequence of teleportations. For CV-flow however, it is not obvious how to go about it, for instance one does not directly have an obvious path cover. We follow the ideas of \cite{miyazaki_analysis_2015}, associating open graphs with CV-flow to equivalent open graphs with causal flow which allows circuit extraction, albeit using quite different proof methods.

More precisely, for the general CV-flow case, we show there are totally ordered partitions of
\(O^\mathsf{c}\) called \emph{layer decompositions} such that all the vertices in
each layer can be measured simultaneously and corrected for in
one step as in lemma \ref{lem:cv-flow-multilinear}---that is, there is a CV-flow
from the layer into the remaining unmeasured vertices.
We then show that this ``one-step'' CV-flow can be reduced to a causal flow for
the same layer, at the cost of some additional gates acting on the unmeasured
vertices.
Finally, by repeating this procedure for each layer, we extract a circuit for
the total MBQC procedure, as a sequence of star pattern transformation circuits
and intermediate gates.

\subsubsection*{From CV-flow to causal flow}

%This result is obtained by noting that linear operations on the column space of the correction matrix (i.e. the cut matrix representing the edges between two layers, see definition \ref{def:matrix}) correspond to \emph{controlled corrections} on the graph state of the form
%\begin{equation}
%  \CXo_{j,k}(s) \prod_{\ell \in N(k)} \CZo_{j,l}(s).
%\end{equation}
%The additional \(\CZo\) operations can be interpreted as edges in a new open graph whose associated state is obtained from the original graph state by adding \(s\) to the weight of every edge \(j \to \ell\), with an additional \(\CXo(s)\) gate5in this new representation for the two open graphs to represent the same graph state. This gate is represented graphically as a directed edge, see figure~\ref{fig:triangularisation}.

%The fact that the additional controlled gates act only on the outputs is crucial: it allows us reduce \(\Oo_\eta\) to a sequence of single-gate teleportation operations. Since the \(\CXo\) gates never appear in between a measurement and the corresponding \(\CZo\) gate for the teleportation, nor do they act on the auxiliary squeezed states before they are consumed in the teleportation, the projections can be brought forward and the squeezed inputs delayed to obtain a single gate teleportation circuit within the larger circuit of the MBQC procedure.

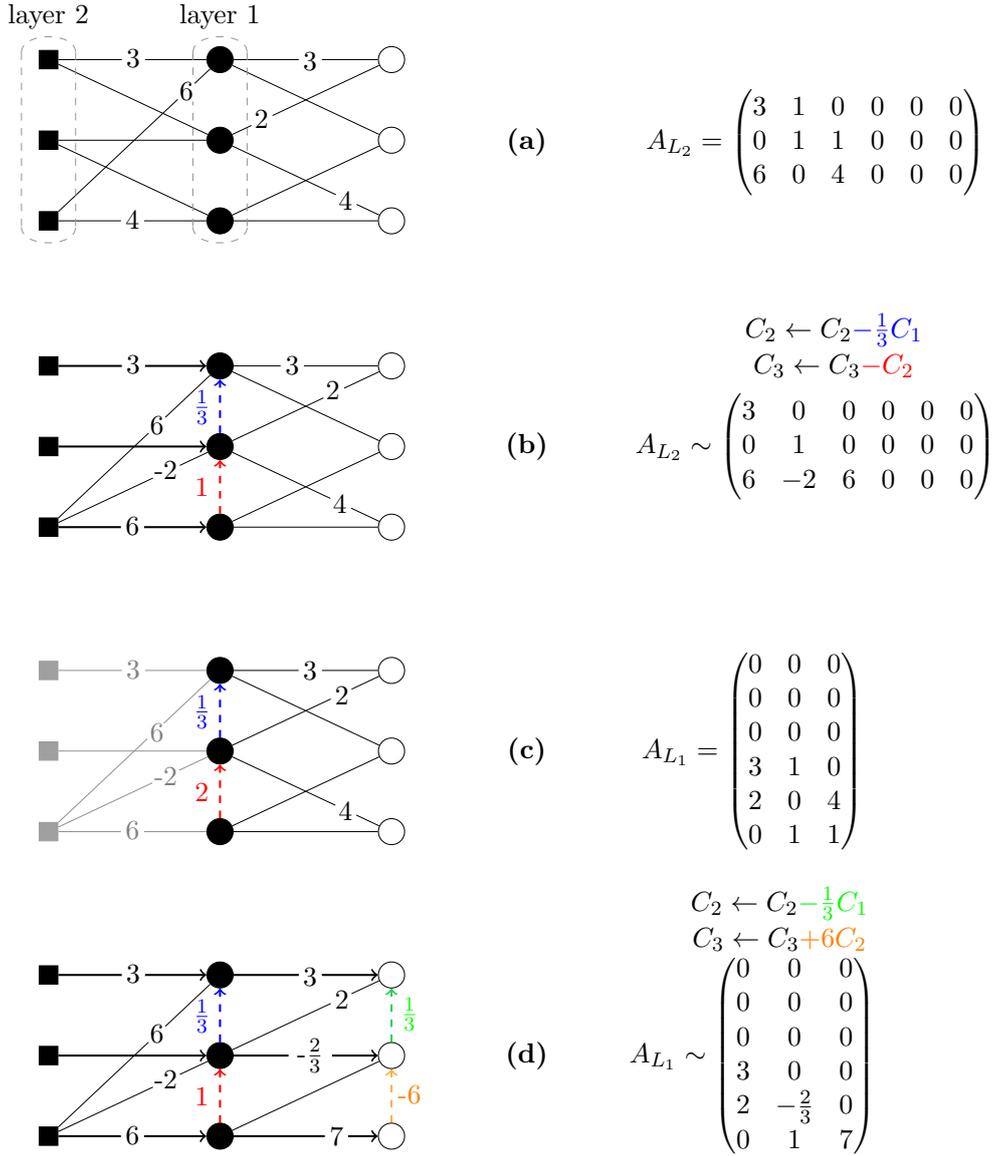
\begin{figure}
  \centering
  \begin{center}
    \scalebox{.95}{\begin{tikzpicture}
	\begin{pgfonlayer}{nodelayer}
		\node [style={black_square}] (0) at (-13.25, 6.05) {};
		\node [style={black_square}] (1) at (-13.25, 3.8) {};
		\node [style={black_square}] (2) at (-13.25, 1.55) {};
		\node [style={black_circle}] (3) at (-8.5, 6.05) {};
		\node [style={black_circle}] (4) at (-8.5, 3.8) {};
		\node [style={black_circle}] (5) at (-8.5, 1.55) {};
		\node [style=none] (9) at (8, 3.75) {\(A_{L_2} = \begin{pmatrix} 3 & 1 & 0 & 0 & 0 & 0 \\ 0 & 1 & 1 & 0 & 0 & 0 \\ 6 & 0 & 4 & 0 & 0 & 0 \end{pmatrix}\)};
		\node [style=none] (29) at (8, -4.75) {\(A_{L_2} \sim \begin{pmatrix} 3 & 0 & 0 & 0 & 0 & 0 \\ 0 & 1 & 0 & 0 & 0 & 0\\ 6 & -2 & 6 & 0 & 0 & 0  \end{pmatrix}\)};
		\node [style=none] (30) at (8.5, -1.5) {\(C_2 \leftarrow C_2 {\color{blue} - \frac{1}{3} C_1}\)};
		\node [style=none] (31) at (-9, -3.65) {\({\color{blue} \frac{1}{3}}\)};
		\node [style={white_circle}] (34) at (-3.75, 3.8) {};
		\node [style={white_circle}] (35) at (-3.75, 6.05) {};
		\node [style=blank] (36) at (-10.9, 6.075) {3};
		\node [style=blank] (37) at (-9.425, 5.15) {6};
		\node [style=blank] (38) at (-10.925, 1.575) {4};
		\node [style=blank] (39) at (-7.35, 4.375) {2};
		\node [style={white_circle}] (40) at (-3.75, 1.55) {};
		\node [style=blank] (41) at (-6, 6) {3};
		\node [style=blank] (42) at (-5.025, 2.1) {4};
		\node [style={black_square}] (52) at (-13.25, -2.5) {};
		\node [style={black_square}] (53) at (-13.25, -4.75) {};
		\node [style={black_square}] (54) at (-13.25, -7) {};
		\node [style={black_circle}] (55) at (-8.5, -2.5) {};
		\node [style={black_circle}] (56) at (-8.5, -4.75) {};
		\node [style={black_circle}] (57) at (-8.5, -7) {};
		\node [style={white_circle}] (58) at (-3.75, -4.75) {};
		\node [style={white_circle}] (59) at (-3.75, -2.5) {};
		\node [style=blank] (60) at (-10.9, -2.475) {3};
		\node [style=blank] (62) at (-10.925, -6.975) {6};
		\node [style=blank] (63) at (-5.375, -3.2) {2};
		\node [style={white_circle}] (64) at (-3.75, -7) {};
		\node [style=blank] (65) at (-6.5, -2.5) {3};
		\node [style=blank] (66) at (-5.175, -6.375) {4};
		\node [style=none] (67) at (-9, -5.9) {\({\color{red} 1}\)};
		\node [style=none] (70) at (6.25, -13.25) {\(A_{L_1} = \begin{pmatrix} 0 & 0 & 0 \\ 0 & 0 & 0 \\ 0 & 0 & 0 \\ 3 & 1 & 0 \\ 2 & 0 & 4 \\ 0 & 1 & 1 \end{pmatrix}\)};
		\node [style=none] (71) at (6.25, -21.75) {\(A_{L_1} \sim \begin{pmatrix} 0 & 0 & 0 \\ 0 & 0 & 0 \\ 0 & 0 & 0 \\ 3 & 0 & 0 \\ 2 & -\frac{2}{3} & 0 \\ 0 & 1 & 7 \end{pmatrix}\)};
		\node [style=grey square] (78) at (-13.25, -11) {};
		\node [style=grey square] (79) at (-13.25, -13.25) {};
		\node [style=grey square] (80) at (-13.25, -15.5) {};
		\node [style={black_circle}] (81) at (-8.5, -11) {};
		\node [style={black_circle}] (82) at (-8.5, -13.25) {};
		\node [style={black_circle}] (83) at (-8.5, -15.5) {};
		\node [style={white_circle}] (84) at (-3.75, -13.25) {};
		\node [style={white_circle}] (85) at (-3.75, -11) {};
		\node [style=blank] (86) at (-10.9, -10.975) {{\color{gray} 3}};
		\node [style=blank] (88) at (-10.925, -15.475) {{\color{gray} 6}};
		\node [style=blank] (89) at (-5.125, -11.7) {2};
		\node [style={white_circle}] (90) at (-3.75, -15.5) {};
		\node [style=blank] (91) at (-6, -11) {3};
		\node [style=blank] (92) at (-5.025, -14.925) {4};
		\node [style=none] (93) at (-9, -14.4) {\({\color{red} 2}\)};
		\node [style={white_circle}] (103) at (-3.75, -19.5) {};
		\node [style={white_circle}] (104) at (-3.75, -21.75) {};
		\node [style={white_circle}] (109) at (-3.75, -24) {};
		\node [style=none] (113) at (-3.25, -20.6) {\({\color{green} \frac{1}{3}}\)};
		\node [style=none] (114) at (-3.25, -22.85) {{\color{orange}-6}};
		\node [style=none] (117) at (-14, 6.25) {};
		\node [style=none] (118) at (-12.5, 6.25) {};
		\node [style=none] (119) at (-14, 1.375) {};
		\node [style=none] (120) at (-12.5, 1.375) {};
		\node [style=none] (121) at (-9.25, 6.25) {};
		\node [style=none] (122) at (-7.75, 6.25) {};
		\node [style=none] (123) at (-9.25, 1.375) {};
		\node [style=none] (124) at (-7.75, 1.375) {};
		\node [style=none] (125) at (-13.25, 7.25) {layer 2};
		\node [style=none] (126) at (-8.5, 7.25) {layer 1};
		\node [style=none] (127) at (0, 3.75) {\textbf{(a)}};
		\node [style=none] (128) at (0, -4.75) {\textbf{(b)}};
		\node [style=none] (129) at (0, -13.25) {\textbf{(c)}};
		\node [style=none] (130) at (0, -21.75) {\textbf{(d)}};
		\node [style=none] (131) at (8.5, -2.5) {\(C_3 \leftarrow C_3 {\color{red} - C_2}\)};
		\node [style=none] (132) at (-9, -12.15) {\({\color{blue} \frac{1}{3}}\)};
		\node [style=blank] (134) at (-10.25, -4.15) {6};
		\node [style=blank] (135) at (-10, -5.425) {-2};
		\node [style=blank] (136) at (-10.25, -12.725) {{\color{gray} 6}};
		\node [style=blank] (137) at (-10, -13.975) {{\color{gray} -2}};
		\node [style=none] (138) at (-9, -20.65) {\({\color{blue} \frac{1}{3}}\)};
		\node [style={black_square}] (139) at (-13.25, -19.5) {};
		\node [style={black_square}] (140) at (-13.25, -21.75) {};
		\node [style={black_square}] (141) at (-13.25, -24) {};
		\node [style={black_circle}] (142) at (-8.5, -19.5) {};
		\node [style={black_circle}] (143) at (-8.5, -21.75) {};
		\node [style={black_circle}] (144) at (-8.5, -24) {};
		\node [style=blank] (145) at (-10.9, -19.475) {3};
		\node [style=blank] (146) at (-10.925, -23.975) {6};
		\node [style=none] (147) at (-9, -22.9) {\({\color{red} 1}\)};
		\node [style=blank] (148) at (-10.25, -21.15) {6};
		\node [style=blank] (149) at (-10, -22.425) {-2};
		\node [style=none] (150) at (7, -17.5) {\(C_2 \leftarrow C_2 {\color{green} - \frac{1}{3} C_1}\)};
		\node [style=none] (151) at (7, -18.5) {\(C_3 \leftarrow C_3 {\color{orange} + 6 C_2}\)};
		\node [style=blank] (152) at (-5.125, -20.2) {2};
		\node [style=blank] (153) at (-6, -19.55) {3};
		\node [style=blank] (154) at (-5.275, -24) {7};
		\node [style=blank] (155) at (-6, -21.75) {-$\frac{2}{3}$};
	\end{pgfonlayer}
	\begin{pgfonlayer}{edgelayer}
		\draw (3) to (35);
		\draw (0) to (3);
		\draw (3) to (2);
		\draw (0) to (4);
		\draw (4) to (1);
		\draw (5) to (1);
		\draw (4) to (35);
		\draw (2) to (5);
		\draw (4) to (40);
		\draw (5) to (40);
		\draw (5) to (34);
		\draw (34) to (3);
		\draw (55) to (59);
		\draw [style={black_edge}] (52) to (55);
		\draw [style={black_thick_l}] (56) to (53);
		\draw (56) to (59);
		\draw [style={black_edge}] (54) to (57);
		\draw (56) to (64);
		\draw (57) to (64);
		\draw (57) to (58);
		\draw (58) to (55);
		\draw [style={red_dashed_arrow_right}] (57) to (56);
		\draw (81) to (85);
		\draw [style=new edge style 0] (78) to (81);
		\draw [style=new edge style 0] (82) to (79);
		\draw (82) to (85);
		\draw [style=new edge style 0] (80) to (83);
		\draw (82) to (90);
		\draw (83) to (90);
		\draw (83) to (84);
		\draw (84) to (81);
		\draw [style={red_dashed_arrow_right}] (83) to (82);
		\draw [style=new edge style 1] (117.center) to (119.center);
		\draw [style=new edge style 1] (118.center) to (120.center);
		\draw [style=new edge style 1, bend right=90] (119.center) to (120.center);
		\draw [style=new edge style 1, bend left=90] (117.center) to (118.center);
		\draw [style=new edge style 1] (121.center) to (123.center);
		\draw [style=new edge style 1] (122.center) to (124.center);
		\draw [style=new edge style 1, bend right=90] (123.center) to (124.center);
		\draw [style=new edge style 1, bend left=90] (121.center) to (122.center);
		\draw [style={green_edge}] (104) to (103);
		\draw [style={blue_dashed_arrow_right}] (56) to (55);
		\draw [style={blue_dashed_arrow_right}] (82) to (81);
		\draw (55) to (54);
		\draw (56) to (54);
		\draw [style=new edge style 0] (82) to (80);
		\draw [style=new edge style 0] (81) to (80);
		\draw [style={black_edge}] (139) to (142);
		\draw [style={black_thick_l}] (143) to (140);
		\draw [style={black_edge}] (141) to (144);
		\draw [style={red_dashed_arrow_right}] (144) to (143);
		\draw [style={blue_dashed_arrow_right}] (143) to (142);
		\draw (142) to (141);
		\draw (143) to (141);
		\draw [style={orange_edge}] (109) to (104);
		\draw [style={black_thick_l}] (103) to (142);
		\draw (143) to (103);
		\draw (104) to (144);
		\draw [style={black_thick_l}] (104) to (143);
		\draw [style={black_thick_l}] (109) to (144);
	\end{pgfonlayer}
\end{tikzpicture}}
  \end{center}
  \caption{An example of the reduction of CV-flow to causal flow by the
    triangularisation procedure as described in section \ref{sec:cv-linear},
    using the graph with CV-flow from figure \ref{fig:CV_flow}.
    We represent the open graphs in the sequential steps of the procedure
    (left) alongside the corresponding cut matrix for the layer in consideration
    at that step (right, layer 2 for \textbf{(a)},\textbf{(b)} and layer 1 for
    \textbf{(c)},\textbf{(d)}) and the column-space operations which are made on
    the cut matrices throughout.
    Starting from an open graph with CV-flow in two layers \textbf{(a)}, an upper
    triangularisation of the cut matrix for the first layer gives a new open
    graph where there is now a causal flow from the vertices in layer 2 into the
    vertices in layer 1 \textbf{(b)}
    Since permuting columns of the correction matrix is a ``free'' operation (it
    corresponds to relabelling unmeasured nodes), the matrix has the form of
    lemma \ref{lem:triangular_form}.
    This comes at the cost of additional weighted \(\CXo\) gates acting in layer
    1 (lemma \ref{lem:triangularisation}), represented as colored edges directed
    from the source of the \(\CXo\) to its target.
    These \(\CXo\) gates in the graph state are colored-matched to the term in
    the matrix triangularisation from which they come.
    We repeat this procedure for layer 1, where the vertices in layer 2 cannot be
    used for corrections since they come before layer 1 in the CV-flow order
    \textbf{(c)}, resulting in a final graph state where the causal flow
    has been indicated by the bold arrows \textbf{(d)} directed as \(u \to
    f(u)\) for each vertex \(u \in O^\mathsf{c}\).
    As per lemma \ref{lem:cv_cover} the edpoints of causal flow edges from
    different layers in the resulting open graph match up, resulting in a path
    cover.
    % For this graphical notation to make formal sense, the \(\CXo\) gates from
    % the triangularisation of layer 2 must appear in the circuit form of the
    % graph state after all the \(\CZo\) gates acting in layer 2, but before any
    % \(\CZo\) gate acting only in subsequent layers.
  }
  \label{fig:triangularisation}
\end{figure}

% \begin{figure}
%   \centering
%   \begin{center}
%     \scalebox{1.1}{\tikzfig{figures/examples/extract-ex}}
%   \end{center}
%   \caption{An example of the reduction of CV-flow to causal flow by the
%     triangularisation procedure.
%     Starting from a graph with CV flow (see figure~\ref{fig:CL_flow}), an upper
%     triangularisation of the correction matrix gives a graph with causal flow
%     (bottom right), up to additional
%     weighted \(\CXo\) gates acting on the outputs (directed edges).}
% \end{figure}
We begin by defining the decomposition into layers as follows.

\begin{Def}
  \label{def:layer_decomposition}
  Let \((G,I,O)\) be a graph with CV-flow \((c, \prec)\).
  A corresponding \textbf{layer decomposition} of \((G,I,O)\) is a partition
  \(\{L_k\}_{k = 1}^N\) of \(O^\mathsf{c}\) such that if \(i \in L_m\), \(j \in
  L_n\) and \(i \prec j\) then \(n < m\) (as elements of \(\NN\)).
\end{Def}

It is straightforward to see that if \(\{L_k\}\) is such a decomposition, we can
measure the layers in the order given by \(L_{k+1} \prec L_k\) and  corrections
are always possible.
We chose this inverted ordering for layers (from the last layer measured to the
first) in accordance with \cite{mhalla_finding_2008}.
It is equally easy to see that in this procedure, we can measure all the vertices
in each layer before applying any correction for measurements errors in the
layer.

In order to extract causal flows from CV-flows, we need a matricial
characterisation of causal flow:
\begin{Lem}[Matrix form of causal flow]
  Let \((G,I,O)\) be an open graph with CV-flow for a total measurement
  order \(<\), and \(L \subseteq O^\mathsf{c}\) a subset of
  vertices.\footnote{Intuitively, \(L\) can be viewed as a set of unmeasured
    vertices for which we are looking for a causal flow.}
  Then there is a subset \(C \subseteq P(L)^\mathsf{c}\) with \(\abs{L} =
  \abs{C}\) and a causal flow \(L \to C\) if and only if the correction matrix
  \(A_L^<\) of \(L\) can be written as
  \begin{equation}
    \label{eq:simple-triang}
    A_L^< = M \cdot \begin{pmatrix}
      X & 0 \\
      Y & T
    \end{pmatrix} \cdot N
  \end{equation}
  where \(M\) and \(N\) are permutation matrices, \(T\) is a lower triangular
  \(\abs{V} \times \abs{C}\) matrix with non-zero diagonal and \(X, Y\) are
  arbitrary real matrices.
  In other words, we can turn \(A_L^<\) into the partial triangular form of
  equation \eqref{eq:simple-triang} only by reordering rows and columns, which
  in turn corresponds to relabelling the vertices of the graph \(G\).
  \label{lem:triangular_form}
\end{Lem}
\begin{proof}
  \((\impliedby)\) If \(A_L\) takes the form described, then the diagonal elements
  of \(T\) determine a single correction vertex in \(C\) for each vertex in \(L\),
  as well as a measurement order such that there is no back-action: the order of
  the columns in \(T\) (since all elements above the diagonal are now 0).
  Thus, there is a causal flow \(L \to C\).

  \((\implies)\) If there is a causal flow \(L \to C\), then there is a
  measurement order \(<\) on \(L\) such that when measuring vertex \(i \in L\),
  there is a single unmeasured vertex \(j \in I^\mathsf{c}\) to correct onto, and
  this correction has no back-action on previously measured vertices.
  But this implies that if we reorder the columns of \(j\) according to \(<\),
  column \(i\) has only zeros above row \(j\) (otherwise there is a
  back-action), and a non-zero entry in row \(j\) (otherwise it is not possible
  to correct onto \(j\)).
  Repeating this process for each vertex in \(L\) gives \(\abs{L}\) such columns,
  let \(C\) be the corresponding correction vertices.

  Now, \(<\) induces an order \(<_C\) on \(C\) by the causal flow matching.
  Extend \(<\) by letting all previously measured vertices in \(P(L)\) be less than
  \(L\), and \(<_C\) by letting all unmeasured vertices in \(P(L)^\mathsf{c}\) be less
  than \(C\).
  Then, ordering the columns and rows of \(A_L\) according to \(<\) and \(<_C\),
  respectively, results in a matrix of the form described.
\end{proof}

% \begin{Lem}
%   Let \((G,I,O)\) be an open graph with CV-flow, \(V \subseteq O^\mathsf{c}\)
%   a set of unmeasured vertices and \(M \subseteq O^\mathsf{c}\) the set of
%   previously measured vertices.
%   Further assume there is a subset \(C \subseteq M^\mathsf{c}\) such that there
%   is a CV-flow \(V \to C\), then \(\abs{V} < \abs{C}\).
% \end{Lem}

This characterisation of causal flow is the key difference between our proof
method and that of \textcite{miyazaki_analysis_2015}---where they use
arguments based on local complementation to find a causal from a g-flow, we
solve the comparatively easier problem of proving it is always possible to map
an open state with CV-flow to one where the correction matrix takes this form.

The approach now is, having broken the measurement pattern down into layers, we show that the graph over each pair of layers can be seen as having flow, by transforming the correction matrix such that it takes the above triangular form. 
Reordering rows and columns of the correction matrix simply corresponds to relabelling of the vertices, however, we will also require linear addition of columns.
This matrix or graphical operation, it turns out, is physically equivalent to applying \(\CXo\) gates, which are exactly the additional operations in the equivalence we mentioned.

This emerges from the following stabiliser condition for controlled operators,
which is approximate for any finite squeezing and, as for the standard
stabiliser conditions, only holds perfectly in the infinite squeezing limit.
We need to extend this result beyond just graph states in order to include
the states which occur in the MBQC, i.e.\, those of the form
\(\mathscr{G}_\eta[\rho]\) (equation~\eqref{eq:g_eta_def}). These are
essentially graph states but in which some of the auxiliary states \(g_\eta\)
are replaced by inputs. 

Then, the afore-mentioned result on stabilisers takes the form:
\begin{restatable}[Approximate controlled stabilizers]{Lem}{ApproxStab}
  \label{lem:cont-stab}
  Let \((G,I,O)\) be an open graph, \(j \in G\) and \(k \in I^\mathsf{c}\).
  Then, for any physical input state \(\rho \in D(\Hcal^{\otimes \abs{I}})\) and
  \(s \in \RR\),
  \begin{equation}
    \lim_{\eta \to \infty} \norm{ \Uch_{\CXo_{j,k}(s) \CZo_{j,N(k)}(s)} \circ \Gch_\eta[\rho] - \Gch_\eta[\rho]} = 0.
  \end{equation}
\end{restatable}
\begin{proof}
  The proof is left to appendix \ref{sapp:approx_stab}.
\end{proof}

In this way, the action of specific \(\CZo\) operations---which are what are
used to create or remove edges in the graph---are equivalent (in the infinite
squeezing limit) to the application of a \(\CXo\) operation, since we have
\begin{equation}
  \Uch_{\CXo_{j,k}(-s)} \circ \Gch_\eta[\rho] \approx \Uch_{\CZo_{j,N(k)}(s)} \circ \Gch_\eta[\rho].
\end{equation}
This allows us to
achieve the our goal:

\begin{Prop}[Triangularisation]
  If \((G,I,O)\) is an open graph with CV-flow, and \(L_1\) is the last layer in
  a corresponding layer decomposition \(\{L_k\}_{k = 1}^N\)
  (definition~\ref{def:layer_decomposition}), then \((G,I,O)\) is approximately
  equivalent to an open graph with a causal flow \(L_1 \to O\), up to weighted
  \(\CXo\) gates acting in \(O\) and reordering the vertices in \(L_1\).
  \label{lem:triangularisation}
\end{Prop}
\begin{proof}
  Let \(A_{L_1}\) be the correction matrix of \(L_1\) for a given CV-flow order.
  Then, we can reorder the columns of \(A_{L_1}\) by relabeling the unmeasured
  vertices, and we can reorder the rows of \(A_{L_1}\) by choosing a different
  measurement order for vertices in \(L_1\).

  Further let \(j, k \in O\), then by lemma \ref{lem:cont-stab} applying the gate
  \(\CXo_{j,k}(-s)\) on the graph state induces new edges in the graph state 
  \emph{in the infinite squeezing limit}.
  The result on the correction matrix is the transformation
  \begin{equation}
    C_j \longmapsto C_j + s C_k,
  \end{equation}
  where \(C_j\) is the \(j\)-th column of \(A_{L_1}\).

  By the definition of CV-flow, for each \(v \in L\) we have that
  \begin{equation}
    A_{L_1}^< c_v(1_v) = 1_v,
  \end{equation}
  so that \(c_v(1_v)\) gives a sum of columns \(A\) which contains a single \(1\)
  in the row corresponding to \(v\).
  Repeating this for each \(v \in L_1\), we obtain \(\abs{L_1}\) such columns, each
  with the \(1\) on a different row, so that by reordering rows and columns we
  can write \(A_{L_1}\) as
  \begin{equation}
    A_{L_1} \sim \begin{pmatrix}
      X & 0 \\
      Y & I_{\abs{L_1}}
    \end{pmatrix}
  \end{equation}
  where \(I_{\abs{L_1}}\) is the \(\abs{L_1} \times \abs{L_1}\) identity matrix.
  Then, \(A_{L_1}\) takes the form described in lemma \ref{lem:triangular_form}, and
  this partial triangularisation procedure corresponds to extracting additional
  \(\CXo\) gates from the graph as described above.
  Then, the open graph \((G,I,O)\) is approximately equivalent to a graph with
  causal flow \(L_1 \to O\), up to \(\CXo\) gates acting in \(O\).
\end{proof}

The fact that the additional controlled gates act only on the outputs is
crucial: it will allow us reduce the total physical map to a sequence of single-gate
teleportation operations.
Since the \(\CXo\) gates never appear in between a measurement and the
corresponding \(\CZo\) gate for the teleportation, nor do they act on the
auxiliary squeezed states before they are consumed in the teleportation, the
projective measurements can be brought forward and the squeezed inputs delayed
to obtain a single gate teleportation circuit within the larger circuit
representing the total physical map of the computation.

\subsubsection*{Path cover of CV-flow}

Now, using these two lemmas, we obtain a causal flow from the last layer \(L_1\) of a
decomposition into a subset of the outputs by adding \(\CXo\) gates.
Most importantly, this subset is then only connected to \(L_1\) so it can be
removed from the open graph as far as determining flows on the remainder is
concerned.
As a result, we can reduce a graph to a sequence of causal flows by peeling off
each layer one-by-one.
\begin{Lem}[Graph reduction]
  If \((G,I,O)\) is an open graph with CV-flow \((\prec,
  (c_j)_{j=1}^N)\), and corresponding layer decomposition
  \(\{L_k\}_{k=1}^N\) then there is \(C_1 \subseteq O\) such that there is a
  causal flow \(L_1 \to C_1\) with \(\abs{L_1} = \abs{C_1}\), up to a product
  \(\To\) of weighted \(\CXo\) gates acting in \(O\).
  Furthermore, let \(G'\) be the graph state obtained from the triangularisation
  procedure for layer \(L_1\), then \(\left( G' \setminus C_1, I \setminus C_1,
  L_1 \cup (O \setminus C_1) \right)\) has CV-flow \((\prec, (c_j)_{j=2}^N)\) and
  layer decomposition \(\{L_k\}_{k=2}^N\).
  \label{lem:CV->simple}
\end{Lem}
\vspace{-5mm}
\begin{proof}
  The first part follows straightforwardly from lemmas \ref{lem:triangular_form}
  and \ref{lem:triangularisation}.
  The second is immediate once one realises the following: by the third
  condition in the definition of causal flow, if there is a causal flow \(C_1
  \to L_1\), \(C_1\) cannot be connected to any vertex in a layer \(k > 1\).
  Since \(L_1\) is measured last, so \(C_1\) must be connected \emph{only} to
  \(L_1\) (and possibly \(O\)).
  As a result, we can remove \(C_1\) from the graph for subsequent layers: since
  it is not connected to any previous layer \(k > 1\), it never appears in any
  subsequent correction subgraphs.
  As a result, the truncated CV-flow and layer decomposition remain valid for
  the reduced graph.
\end{proof}

This ``peeling'' procedure also allows us to determine a path cover of \((G, I,
O)\), by noting that each layer causal flow has a path cover, and the endpoints
of each of these covers meet up.
So, by a successive applications of this lemma, we obtain the final ingredient
to our proof, a CV-flow analogue of lemma \ref{lem:simple-cover}:
\begin{Lem}[CV-flow path cover]
  Let \((G,I,O)\) be an open graph with CV-flow, then there is a path cover of
  \((G,I,O)\) whose edges are causal flow edges of the triangularised graph
  (\ref{lem:triangularisation}).
  If \(\abs{I} = \abs{O}\), every path is indexed by an input.
  \label{lem:cv_cover}
\end{Lem}
\begin{proof}
  Let \(\{L_k\}\) be a layer decomposition of \((G,I,O)\), and consider each
  vertex \(j \in O\) the endpoint of a path.
  Then, by lemma \ref{lem:CV->simple} there is \(C_1 \subseteq O\) such that
  there is a causal flow and a bijection \(L_1 \to C_1\); label each vertex in
  \(L_1\) by its image under the causal flow matching.
  Then, remove \(C_1\) from the graph as in lemma \ref{lem:CV->simple}, and
  repeat the process.
  Since \(\bigcup_{k=1}^N L_k \cup O = G\), we eventually label the whole graph.
  Furthermore, the resulting paths never cross: if they did, there would be two
  vertices in the same layer corrected onto the same node--but this is impossible,
  by the definition of causal flow.
  Thus the resulting set of paths is a path cover for \((G,I,O)\).

  Finally, if \(\abs{I} = \abs{O}\), every input is the beginning of some path,
  since we measure all \(j \in I\) but can never correct onto \(I\).
  Since there are exactly \(\abs{O} = \abs{I}\) paths, every path must begin in
  \(I\) and end in \(O\), and every path is indexed by an input.
\end{proof}

As a corollary, we obtain bounds on the number of inputs and outputs of an open
graph if it has a CV-flow:
\begin{Cor}
  Let \((G,I,O)\) be an open graph with CV-flow, then \(\abs{I} \leqslant
  \abs{O}\).
\end{Cor}

\begin{proof}
  By the proof to the lemma, every input is the beginning of a path that ends in
  \(O\), and these paths never cross, such that even their endpoints in \(O\)
  cannot coincide.
  Then, the collection of paths describes an injection \(I \to O\), since each
  path uniquely associates an endpoint in \(O\) to each input.
\end{proof}

\begin{Cor}
  Let \((G,I,O)\) be an non-empty open graph with CV-flow, then \(\abs{O} 
  \neq 0\).
\end{Cor}

\begin{proof}
  If \(G \neq \varnothing\) and the open graph has CV-flow, then there is a
  path cover of the open graph which contains at least one path.
  This path must end at an output vertex, thus \(\abs{O} \neq 0\).
\end{proof}

\subsubsection*{Putting it all together}

\begin{figure}
  \centering
  \begin{center}
    \scalebox{1.15}{\input{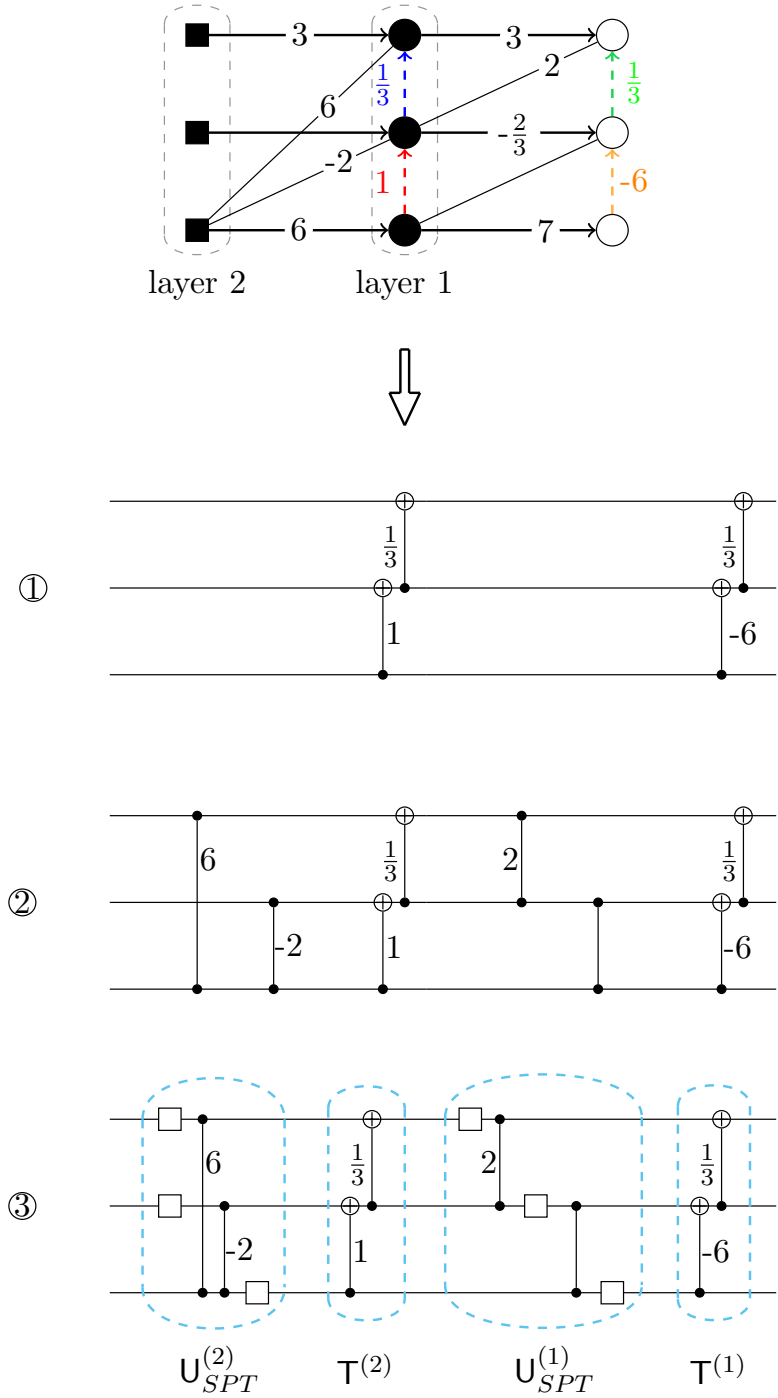}}
  \end{center}
  \caption{Example of the final circuit extraction procedure for an open graph
    with CV-flow.
    We use the open graph state resulting from the triangularisation (figure
    \ref{fig:triangularisation}) of the graph state with CV-flow from figure
    \ref{fig:CV_flow}.
    (1) The causal flow edges (bold directed edges) of the open graph are
    interpreted as edges in a circuit, and the \(\CXo\) gates (colored, dashed,
    directed edges) from the triangularisations of each layer are added.
    These are ordered by layer, and within a layer by (the inverse of) the order of column space
    operations in the triangularisation of the cut matrix for that layer.
    They separate the final circuit into individual sections, with each section
    corresponding to a single layer of the open graph.
    (2) The auxiliary (non-causal flow) edges for each section are added to the
    circuit as \(\CZo\) gates.
    (3) Each of the causal flow edges corresponds to a \(\Jo\) gate
    (equation \eqref{eq:j-gate}) in the final circuit.
    Adding these gates completes the SPT for each layer, and we identify the
    different sections corresponding to the decomposition of theorem
    \ref{thm:CV-circuit} in the final circuit.
  }
  \label{fig:CV-circuit}
\end{figure}

We are finally ready to state our main convergence result (the proof is in
appendix \ref{app:convergence}):
\begin{restatable}[CV-flow circuit]{Thm}{CVconvergence}
  \label{thm:CV-circuit}
  If \((G,I,O)\) is an open graph with CV-flow and \(\abs{I} = \abs{O}\)
  then the CV-flow correction protocol converges to a unitary acting on the
  input state.
  If \(\{L_k\}_{k=1}^n\) is a corresponding layer decomposition, for any
  \(\vec{\a}, \vec{\b}, \vec{\g} \in \RR^{\abs{O^\mathsf{c}}}\) let
  \begin{equation}
    \Wo(\vec{\a},\vec{\b},\vec{\g})
    \coloneqq \prod_{k=1}^{n} \To^{(k)} \Uo_{SPT}^{(k)}(\vec{\a},\vec{\b},\vec{\g}),
  \end{equation}
  where \(\Uo_{SPT}^{(k)}\) is the circuit extracted for the \(k\)-th layer
  using the causal flow from lemma \ref{lem:CV->simple}, and \(\To^{(k)}\)
  contains the \(\CXo\) gates obtained from the triangularisation of the
  CV-flow (lemma \ref{lem:triangularisation}).

  Then, for any \(\vec{\a}, \vec{\b}, \vec{\g} \in \RR^{\abs{O^\mathsf{c}}}\)
  and any physical input state \(\rho \in D(\Hcal^{\otimes \abs{I}})\),
  \begin{equation}
    \lim_{\eta \to \infty}
    \Cch_\eta(\vec{\a}, \vec{\b}, \vec{\g})[\rho]
    = \Wo(\vec{\a},\vec{\b},\vec{\g}) \rho \Wo^*(\vec{\a},\vec{\b},\vec{\g}).
  \end{equation}
  \(\To^{(k)} \Uo_{SPT}^{(k)}(\vec{\a}, \vec{\b}, \vec{\g})\) acts on the
  qumodes represented by wires indexed by \(L_k\), and the total product \(\Wo\)
  acts on the qumodes represented by wires indexed by \(I\).
\end{restatable}
\begin{proof}
  The proof is left to appendix~\ref{app:CV_circuit}
\end{proof}

In other words, in the ideal limit, we implement the unitary
\begin{equation}
  \begin{aligned}
    \Wo(\vec{\a}, \vec{\b}, \vec{\g}) :
    \Hcal^{\otimes\abs{I}} &\longrightarrow \Hcal^{\otimes\abs{I}} \\
    f &\longmapsto \prod_{k=1}^{n} \To^{(k)}
    \Uo_{SPT}^{(k)}(\vec{\a},\vec{\b},\vec{\g}) f,
  \end{aligned}
\end{equation}
where convergence of the approximation is once again in the strong sense.

An example of the complete circuit extraction procedure for an open graph based
on a CV-flow and using the triangularised causal flow from example
\ref{fig:triangularisation} is given in figure \ref{fig:CV-circuit}.

\section{Flow for MBQC with qudit graph states}
\label{sec:qudit}
In this section, we give a brief overview of how our methods apply to the qudit case.
We leave a more in-depth discussion for future work, as some results are missing
when comparing with the qubit case \cite{browne_generalized_2007}. Some of these
questions are solved in our follow-up article \cite{booth_outcome_2021}.

\subsection{Computational model}
For qudit MBQC, the setup is very analogous to CV: letting \(d\) be an odd
prime, the state space is also a space of square integrable functions \(\Hch =
\mathrm{L}^2(\ZZ_d,\CC)\) with respect to a discrete measure, consisting of
functions \(\ZZ_d \to \CC\) with norm
\begin{equation}
  \qq{for any \(\psi \in \Hch\),} \norm{\psi} \coloneqq \sum_{n=1}^d \abs{\psi(n)}^2.
\end{equation}
Vectors in \(\Hch\) can be represented with respect to an orthonormal basis
\(\{\ket{n}\}_{n=1}^d\), and we adopt the bra-ket notation which is better
suited to this case than CV:
\begin{equation}
  \qq{for any \(\psi \in \Hch\),} \ket{\psi} \coloneqq \sum_{n = 1}^d \psi(n) \ket{n}.
\end{equation}
This space also comes equipped with two operators \(\Xo\) and \(\Zo\). Writing
\(m + n\) and \(m \cdot n\) the addition and multiplication in the finite field
\(\ZZ_d\) (that is, the usual operations modulo \(d\)), we have
\begin{align}
  \Xo \ket{n} \coloneqq \ket{n + 1}, \\
  \Zo \ket{n} \coloneqq e^{i\frac{2 \pi n}{d}} \ket{n}, \\
  \Xo^d = \Io = \Zo^d, \\
  \Zo \Xo = e^{i \frac{2 \pi}{d}} \Xo \Zo.
\end{align}
They are also linked by the Hadamard or discrete Fourier transform \(\Ho\),
\begin{equation}
  \Ho \ket{n} \coloneqq \frac{1}{\sqrt{d}} \sum_{k=1}^d e^{i\frac{2 \pi k n}{d}} \ket{k},
\end{equation}
so that
\begin{equation}
  \Ho \Zo \Ho^{-1} = \Xo \qand \Ho \Xo \Ho^{-1} = \Zo^{-1}.
\end{equation}

Finally, we define the multiplication, controlled-Z and controlled-X gates:
\begin{align}
  \Mo(w) \ket{n} &\coloneqq \ket{w \cdot n}, \\
  \CZo(w) \ket{m} \otimes \ket{n} &\coloneqq e^{i\frac{2 \pi w m n}{d}} \ket{m} \otimes \ket{n}, \\
  \CXo(w) \ket{m} \otimes \ket{n} &\coloneqq \ket{m} \otimes \ket{n + w \cdot m}.
\end{align}

For \(\Mo(w)\) to be unitary, it is necessary for \(w\) to be invertible. In
fact, as we shall see in the following section, for the techniques we developed
for MBQC in the CV case to apply, \(\ZZ_d\) must be a field and therefore \(d\)
must be prime. It is also clear that, algebraically, these operations are
analogous to those used for CV, with the multiplication operator taking the
place of squeezing.

\subsection{Measurement-based quantum computing}
\label{qudit:mbqc}

It is straightforward to see that the usual qubit gate teleportation generalises
to
\begin{center}
  \leavevmode
  \centering
  \Qcircuit @C=1.2em @R=0.4em {
    \lstick{\ket{\phi}} & \qw & \ctrl{2} & \gate{\mathsf{U}}
    & \gate{\Ho} & \meter  & \controlo{2} \cw \\
    & & & \hspace{-1.8cm}w & & & \cwx\\
    \lstick{\ket{0}} & \gate{\Ho} & \control \qw & \qw
    & \qw & \qw &\gate{\Xo^{-wm}} \cwx & \rstick{\Mo(w) \Ho \Uo \ket{\phi}} \qw \\ 
  }
  \vspace{3mm}
\end{center}
where once again \(\Uo\) is any unitary that commutes with \(\CZo(w)\), and we
perform a projective measurement whose elements are projections onto the basis
\(\{\ket{n}\}\). Therefore, as in the CV case, we can interpret the measurement
(including the change-of-basis unitary \(\Uo\)) as a measurement apparatus that
depends on the choice of a unitary \(\Uo\), with the only condition being that
\(\Uo\) must commute with \(\Zo\). It is clear that such a gate teleportation
protocol permits one to implement a large enough range of unitaries to obtain
MBQCs which are universal for qudit quantum computation
\cite{sawicki_universality_2017}.

\subsubsection*{Graph states}

A \textbf{qudit open graph} is an undirected \(\ZZ_d\)-edge-weighted
graph \(G = (V,A)\),\footnote{Once again, \(V\) is the set of vertices of \(G\)
  and \(A\) its \(V \times V\) adjacency matrix with coefficients in \(\ZZ_d\).}
along with two subsets of vertices \(I\) and \(O\), which correspond to the
inputs and outputs of a computation \cite{zhou_quantum_2003}. To this abstract
graph, we associate a physical resource state, the \textbf{graph state}, to be
used in a computation: each vertex \(j\) of the graph corresponds to a single
qudit and thus to a single pair \(\{\Xo_j, \Zo_j\}\).

For a given input state \(\ket{\psi}\) on \(\abs{I}\) qudits, the graph state can
be constructed as follows:
\begin{enumerate}
\item Initialise each non-input qudit, \(j \in I^{\mathsf{c}}\), in the
  auxiliary state \(\ket{\tilde{0}} = \Ho \ket{0}\), resulting in a separable state of the form
  \(\ket{\tilde{0}}^{\otimes \abs{I^{\mathsf{c}}}} \otimes \ket{\psi}\).
\item For each edge in the graph between vertices \(j\) and \(k\) with
  weight \(A_{j,k} \in \ZZ_d\), apply the entangling operation
  \(\Co\Zo_{j,k}\big(A_{j,k}\big)\) between the corresponding qudits.
\end{enumerate}

\subsection{Results}

Now, all of the results proven for continuous variables function analogously in
this setting.
Since the all of the arguments in CV simply extend the convergence proof of
single gate teleportation to the more general MBQC setting, and we have seen
that that proof is replaced by equality in the finite-dimensional case, we can
simply take all the results from previous sections, and drop the convergence
statements for equality.

For the sake of clarity, we repeat some definitions from section~\ref{sec:flow}:
\begin{Def}
  Let \((G,I,O)\) be a qudit open graph, \(<\) a total order over
  the vertices \(V\) of \(G\) and \(A^<\) be the adjacency matrix of \(G\)
  (which therefore has coefficients in \(\ZZ_d\)), such that its columns and
  rows are ordered by \(<\).
  Further, define \(P(j)\) (the ``past'' of vertex \(j\)) as the subset of \(V\)
  such that \(k \in P(j)\) implies \(k \leqslant j\).
  Then the \textbf{correction matrix} \(A_j^<\) of a vertex \(j \in V\) is the 
  matrix \(A_j^< \coloneqq A[P(j), (P(j) \cup I)^\mathsf{c}]\).
  \label{def:matrix}
\end{Def}

The correction matrix tells us how \(\Xo\) and \(\Zo\) operations on unmeasured
vertices are going to affect the previously measured vertices, thus it allows us to
determine how to apply a correction on a specific vertex by controlling this
back-action. The generalised flow condition is now given by:
\begin{Def}
  A qudit open graph \((G,I,O)\) has \textbf{\(\ZZ_d\)-flow} if there exists a
  partial order \(\prec\) on \(O^{\mathsf{c}}\) such that for any total order
  \(<\) that is a linear extension of \(\prec\) and every \(j
  \in O^\mathsf{c}\) there is a function \(c_j : \ZZ_d \to  \ZZ_d^{\abs{G} - j}\) 
  such that for all \(m \in \ZZ_d\) the linear equation
  \begin{equation*}
    A_j^< c_j(m) =
    \begin{pmatrix}
      0 \\ \vdots \\ 0 \\ m
    \end{pmatrix}
    \qq{holds,}
    \tag{\(\star\)}
  \end{equation*}
  where \(A_j^<\) is the correction matrix of vertex \(j\).
  Letting \((c_j)_{j \in O^\mathsf{c}}\) be a such set of functions, we call
  the pair \((\prec, (c_j))\) a \(\ZZ_d\)-flow for \((G, I, O)\).
\end{Def}

Let \((G,I,O)\) be an open \(\ZZ_d\)-graph with \(\ZZ_d\)-flow. Since the gate
teleportation allows us to implement any unitary which commutes with \(\Zo\)
(section~\ref{qudit:mbqc}), further assume a choice of unitary \(\Uo(k)\) is
made for each \(k \in O^\mathsf{c}\) such that \([\Uo(k),\Zo]=0\). Then, by
construction we know the following MBQC protocol is runnable:
\begin{enumerate}
\item measure the non-output vertices in the graph in any order which is a linear
  extension of \(\prec\), in the basis corresponding to \(\Uo\);
  and,
\item immediately after each measurement (say of vertex \(j\)), and before any
  other measurement is performed, correct for the measurement error \(m_j\) by
  applying
  \begin{equation}
    \Co_j^\prec(m_j) \coloneqq \prod_{k \in V \setminus P(j) \cup I}
    \left( \Xo_k^{-c_j(m_j)_k}
      \quad\smashoperator{\prod_{\ell \in N(k) \setminus P(j)}}\quad \Zo_\ell^{-A_{k,\ell} 
      \cdot c_j(m_j)_k} \right),
    \label{eq:correction}
  \end{equation}
  where \(c_j(m_j)_k\) is the \(k\)-th element of the vector \(c_j(m_j) \in 
  \RR^{\abs{V} - \abs{P(j) \cup I}}\).
\end{enumerate}

We write \(\Zch\) the quantum channel corresponding to this protocol. Of course,
the corrections \(\Co_j^\prec\) can be simplified and reduced, as was done in
equation~\eqref{eq:correction_simplified} for the CV case.

\subsection*{Quantum circuit extraction}

Unlike for continuous-variables, for qudits we do not need to worry about
convergence questions, since the gate teleportation protocol holds exactly.
However, it is still of interest to obtain a reversible quantum circuit which
implements the same quantum operations as a given MBQC. Under the assumption of
\(\ZZ_d\)-flow, the methods of section~\ref{sec:circuit} can easily be adapted
to this case, only where each convergence statement is replaced by an equality.

The causal flow condition (definition~\ref{def:flow}) applies to the
qudit case as well, with the same intuition as the qubit and CV cases. Denoting
\(\Fch\) the corresponding quantum channel, we have
\begin{Prop}[Causal flow circuit]
  \label{qudit:causal_flow}
  Suppose the qudit open graph \((G,I,O)\) has a causal flow, then for any
  choice of measurement bases,
  \begin{equation}
    \Fch(\vec{\a}, \vec{\b}, \vec{\g})
    = \Uch_{\Uo_{SPT}},
  \end{equation}
  where \(\Uo_{SPT}\) is the unitary corresponding to the circuit obtained by
  star pattern transformation of \((G,I,O)\).
\end{Prop}
In the star pattern transformation, the gates are given by \(\Jo(i) = \Ho
\Uo(i)\). As before, for a unitary \(\Uo\), \(\Uch_\Uo[\rho] = \Uo \rho \Uo^*\).
\begin{proof}[Sketch of proof]
  As was explained for continuous-variables, causal flow just identifies gate
  teleportation in the corresponding MBQC, which are interspersed with
  intermediate \(\CZo\) gates. Star pattern transformation simply makes this
  intuition formal by providing a way to order the teleportations and
  intermediate \(\CZo\)s to obtain an equivalent circuit. In the CV case, this
  circuit is an approximation since we then need to consider convergence of the
  individual gate teleportations, but for qudits this is no longer necessary so
  we get a proper equality: the teleportation along the edge \(i \to f(i)\)
  exactly implements a gate \(\Jo(i)\).
\end{proof}
\begin{Prop}[\(\ZZ_d\)-flow circuit]
  If \((G, I, O)\) is a qudit open graph with \(\ZZ_d\)-flow and \(\abs{I} = \abs{O}\)
  then the \(\ZZ_d\)-flow correction protocol implements a unitary acting on the
  input state.
  If \(\{V_k\}_{k=1}^n\) is a corresponding layer decomposition, for any
  choice of measurement bases,
  \begin{equation}
    \Zch(\vec{\a}, \vec{\b}, \vec{\g})
    = \Uch_{\prod_{k=1}^{n} \To^{(k)} \Uo_{SPT}^{(k)}},
  \end{equation}
  where \(\Uo_{SPT}^{(k)}\) is the circuit extracted for the \(k\)-th layer
  using the causal flow from the equivalent to lemma \ref{lem:CV->simple}, and
  \(\To^{(k)}\) contains the \(\CXo\) gates obtained from the equivalent to the
  triangularisation of the \(\ZZ_d\)-flow (lemma \ref{lem:triangularisation}).
  \(\To^{(k)} \Uo_{SPT}^{(k)}\) acts on the qudits represented by wires indexed
  by \(V_k\), and the total product acts on the qudits represented by wires
  indexed by \(I\).
\end{Prop}
\begin{proof}[Sketch of proof]
  This proof follows almost entirely the same argument as presented in
  section~\ref{sec:cv-linear}, only where once again all of the limits can be
  taken to be proper equalities. This reduces the general case down to the case
  of causal flows, which can then be extracted using
  proposition~\ref{qudit:causal_flow}.
\end{proof}

To end this section, we make clear why we need to assume that the
operations form a field and thus restrict to only prime dimensions. All
single-qudit gate extractions occur via the gate teleportation protocol, and for
this it suffices for the edge weight \(w\) to have a multiplicative inverse,
i.e.\, it belongs to the group of units \(\mathbb{Z}_d^*\). Then the edge weight
is extracted as a unitary \(\mathsf{M}(w)\), as described in
section~\ref{qudit:mbqc}. So far there is no issue if we allow any dimension but
restrict to only open graphs with edge-weights in the group of units
\(\mathbb{Z}_d^*\).

The problem if we are not operating over a field arises during the extraction of
a more general MBQC, in particular at the triangularisation step. At this point,
we need to triangularise a submatrix of the adjacency matrix of the graph. Doing
so using a standard technique, say Gaussian elimination, risks us ending up with
edge weights that are no longer invertible. If such an edge belongs to the path
cover used in the extraction, i.e.\, the edge is to be extracted as a gate
teleportation, then the circuit extraction algorithm gets stuck.

\section*{Conclusion}
We have defined a notion of flow for continuous variables and proved that it can
be used to obtain a desired unitary, provided sufficient squeezing resources
are available.
We have also obtained a polynomial algorithm for finding CV-flow and  a circuit extraction scheme, which might allow further
comparison of depth and size complexities between circuit models and MBQC, as has
already been obtained in the DV case, with a preliminary result in this direction showing depth separation in Appendix \ref{app:depth}.
We have not considered the question of convergence rates in terms of the
squeezing resources available nor the precision of measurements.
These are highly dependant on the specific choice of measurements, auxiliary
teleportation states, the topology of the graph, and, most problematically, on
the input state itself.

There are further extensions possible to the flow framework.
In particular, one might consider Hilbert spaces over more general locally compact
rings or fields and these come equipped with a different unitary group of
translations.
It is unclear whether a good notion of flow is possible in these spaces in
general, but examples include the cases considered in this article, which
correspond to \(\RR\) and \(\ZZ_d\).
This is also necessary to extend the theory to qudits of arbitrary (non-prime)
local dimension.
Then, one is interested in general in the case where the edges of the graph are
weighted with elements of a ring \(R\), and the correction equations are solved
in the \(R\)-module \(R^n\).

Finally, we only considered a single choice of measurement ``plane'', where the
original gflow condition accounted for several.
It should be possible to extend our framework to account for more planes using
symplectic relations between different measurement bases.

{\raggedright\printbibliography}

\newpage
\appendix

\section{Convergence proofs}
\label{app:convergence}
We use several different, equivalent formulations of quantum theory for the
different proofs of convergence:
\begin{itemize}
\item the proof of gate teleportation, proposition \ref{prop:tele-circuit} is
  done via the Wigner-Weyl-Moyal phase space formulation of quantum theory,
  which is a standard tool in continuous variables quantum computing;
\item the proof that graph states admit approximate controlled stabilisers,
  lemma \ref{lem:cont-stab} is done in the Hilbert space formulation, since this
  is a statement about pure states;
  \item the proofs of convergences of the flow MBQC procedures in the density
    operator formulation, since these two proofs follow from the other two
    convergence results and a continuity argument for quantum channels on
    \(\mathfrak{T}(\mathcal{H})\).
\end{itemize}

\subsubsection*{The Wigner function}

We briefly review the Wigner formulation.
In fact, we do not need the full phase space picture: it is sufficient for our
purposes to understand how to represent states and measurements on those states.

In the phase space formalism, to each density operator \(\rho \in D(\Hcal)\)
such that \(\rho = \sum_j c_j \rho_{\psi_j}\), we associate a \emph{real-valued}
square-integrable function in a Hilbert space \(\mathrm{L}^2(\RR
\times \RR, \RR)\), called the Wigner function:
\begin{equation}
  W_\rho(q,p) \coloneqq \frac{1}{\pi} \sum_j c_j \int_{\RR} \psi_j^*(q+y) \psi_j(q-y) e^{2ipy} \dd{y}.
\end{equation}
and it is clear that this association is \(\RR\)-linear, \(W_{\rho + \lambda \s}
= W_\rho + \lambda W_\s\).
The norm is given by
\begin{equation}
  \norm{W} \coloneqq \sqrt{\int_{\RR} \int_{\RR} W(x,y)^2 \dd{x}\dd{y}},
\end{equation}
which agrees with the Hilbert-Schmidt norm, 
\(\norm{W_\rho}^2 = \tr(\rho^2)\).

\subsection{Proof of Proposition \ref{prop:tele-circuit}}
\label{sapp:tele-circuit}

We are interested in proving convergence of the quantum map implemented by:
\begin{center}
  \leavevmode
  \centering
  \Qcircuit @C=1.2em @R=0.4em {
    \lstick{\rho} & \ctrl{2} & \gate{\Uo(\a,\b,\g)}
    & \measureD{\Po} & \controlo{2} \cw & \\
    & \hspace{-4mm}w & & & \cwx \\
    \lstick{\s_\eta} & \control \qw & \qw & \qw & \gate{\Xo(-wm)} \cwx &
    \rstick{\Tch_\eta(\a,\b,\g)[\rho]} \qw \\ 
  }
  \vspace{3mm}
\end{center}

\GateTele*

\begin{proof}
  We first note that we can ignore the parameter \(w\) for the \(\CZo(w)\) gate:
  we know that
  \begin{align}
    \CZo(w) &= \exp(iw\Qo_1 \Qo_2) = \exp(i\Qo_1 w\Qo_2) = \exp(i\Qo_1 \So_2(w)^*\Qo_2\So_2(w)) \\
    &= \So_2(w)^*\exp(i\Qo_1 \Qo_2)\So_2(w) = \So_2(w)^*\CZo(1)\So_2(w) = \So_2(-w)\CZo\So_2(w).
  \end{align}
  Thus, the teleportation circuit is equivalent to
  \begin{center}
    \leavevmode \centering \Qcircuit @C=1.2em @R=0.4em { \lstick{\rho} & \qw &
      \ctrl{2} & \gate{\Uo(\a,\b,\g)}
      & \measureD{\Po} & \controlo{2} \cw & \\
      & & & & & \cwx \\
      \lstick{\s_\eta} & \gate{\So(-w)} & \control \qw & \gate{\So(w)} & \qw &
      \gate{\Xo(-wm)} \cwx &
      \qw \\
    } \vspace{3mm}
  \end{center}
  and, commuting the correction with the squeezing operator, to
  \begin{center}
    \leavevmode \centering \Qcircuit @C=1.2em @R=0.4em { \lstick{\rho} &
      \ctrl{2} & \gate{\Uo(\a,\b,\g)}
      & \measureD{\Po} & \controlo{2} \cw & \\
      & & & & \cwx \\
      \lstick{\s_{\eta - w}} & \control \qw & \qw & \qw & \gate{\Xo(-m)} \cwx &
      \qw & \gate{\So(w)} & \qw \\
    } \vspace{3mm}
  \end{center}
  The additional squeezing \(w\) in the auxiliary state will be absorbed into
  the limit, and the final \(\So(w)\) gate can be added at the end since it
  comes after the teleportation (it is unitary thus continuous and preserves
  limits). Since the ``change of basis'' unitary \(\Uo(\a,\b,\g)\) commutes with
  the \(\CZo\) gate, it can be absorbed into the input state which is arbitrary
  by hypothesis. We have therefore reduced the problem to proving convergence of
  the simpler circuit:
  \begin{center}
    \leavevmode \centering \Qcircuit @C=1.2em @R=0.4em { \lstick{\rho} &
      \ctrl{2}
      & \measureD{\Po} & \controlo{2} \cw & \\
      & & & \cwx \\
      \lstick{\s_\eta} & \control \qw & \qw & \gate{\Xo(-m)} \cwx &
      \rstick{\Tch_\eta[\rho]} \qw \\
    } \vspace{3mm}
  \end{center}
  for an arbitrary input \(\rho \in D(\mathcal{H})\). From \cite{gu_quantum_2009} we
  know that the output of this circuit is
  \begin{equation}
    W_{\mathcal{T}_\eta[\rho]}(q,p) = \int_{\RR} G_\eta(q-x) W_\rho(x,p) \dd{x}
    = W_{\Fo \rho \Fo^*} *_1 g_{\frac{1}{\eta}}(q,p),
  \end{equation}
  where \(*_1\) indicates convolution with respect to the first variable. We
  need to bound the trace distance \(\norm{\mathcal{T}_\eta[\rho] - \Fo \rho
    \Fo^*}\). By \cite{alberti_bures_2000,hou_fidelity_2012} we know that for
  any \(\rho,\sigma \in D(\mathcal{H})\),
  \begin{equation}
    \norm{\rho - \sigma} \leqslant \sqrt{1 - \mathrm{F}(\rho,\sigma)},
  \end{equation}
  where \(\mathrm{F}\) is the Ulhmann fidelity \cite{uhlmann_transition_1976}
  which can be calculated for pure states as
  \begin{equation}
    \mathrm{F}(\rho,\sigma) = \frac{1}{\pi} \int_{\RR}\int_{\RR} W_\rho(q,p)
    W_\sigma(q,p) \dd{q}\dd{p}.
  \end{equation}

  Assume \(\rho\) is a pure state, and furthermore that it is the density
  operator of \(\psi \in \mathrm{L}^1(\RR) \cap \LR\), i.e. the projector \(P_\psi\)
  onto the one-dimensional subspace generated by \(\psi\).
  Then we have
  \begin{align}
    1 - \mathrm{F}(\mathcal{T}_\eta[\rho],\Fo \rho \Fo^*)
    &= 1 - \frac{1}{\pi} \int_{\RR} \int_{\RR} W_{\mathcal{T}_\eta[\rho]}(q,p) W_\rho(p,-q) \dd{q}\dd{p} \\
    &= \frac{1}{\pi} \int_{\RR} \int_{\RR} \Big(W_\rho(p,-q)^2 - W_{\mathcal{T}_\eta[\rho]}(q,p) W_\rho(p,-q)\Big) \dd{q}\dd{p} \\
    &= \frac{1}{\pi} \int_{\RR} \int_{\RR} W_\rho(p,-q) \Big(W_\rho(p,-q) - W_{\mathcal{T}_\eta[\rho]}(q,p)\Big) \dd{q}\dd{p}. 
  \end{align}
  since
  \begin{equation}
    \int_{\RR} \int_{\RR} W_\rho(p,-q)^2 \dd{q}\dd{p} = \tr((\Fo \rho \Fo^*)^2) = \tr(\rho^2) = 1.
  \end{equation}
  \(\psi\) being \(\mathrm{L}^1\) implies that its Wigner transform \(W_\rho\)
  is also \(\mathrm{L}^1\) (\cite{de_gosson_symplectic_2006}, proposition 6.43),
  so
  \begin{align}
    1 - \mathrm{F}(\mathcal{T}_\eta[\rho],\Fo \rho \Fo^*)
    &= \abs{1 - \mathrm{F}(\mathcal{T}_\eta[\rho],\Fo \rho \Fo^*)} \\
    &= \abs{\frac{1}{\pi} \int_{\RR} \int_{\RR} W_\rho(p,-q) \Big(W_\rho(p,-q) - W_{\mathcal{T}_\eta[\rho]}(q,p)\Big) \dd{q}\dd{p} } \\
    &\leqslant  \frac{1}{\pi} \int_{\RR} \int_{\RR} \abs{W_\rho(p,-q) \Big(W_\rho(p,-q) - W_{\mathcal{T}_\eta[\rho]}(q,p)\Big)} \dd{q}\dd{p} \\
    &= \frac{1}{\pi} \int_{\RR} \int_{\RR} \abs{W_\rho(p,-q)} \cdot \abs{W_\rho(p,-q) - W_{\mathcal{T}_\eta[\rho]}(q,p)} \dd{q}\dd{p} \label{eq:wrho_prod}\\
    &\leqslant \frac{1}{\pi^2} \int_{\RR} \int_{\RR} \abs{W_\rho(p,-q) - W_{\mathcal{T}_\eta[\rho]}(q,p)} \dd{q}\dd{p} \label{eq:wrho_prod_removed} \\
    &= \frac{1}{\pi^2} \int_{q \in \RR} \int_{p \in \RR} \abs{W_\rho(p,-q) - W_\rho *_1 g_\frac{1}{\eta}(p,-q)} \dd{q}\dd{p},
  \end{align}
  where we have used the inequality \(\abs{W_\rho(p,-q)} \leqslant
  \tfrac{1}{\pi}\) for pure states (\cite{de_gosson_symplectic_2006}, section
  6.4.3) to go from equation \eqref{eq:wrho_prod} to
  \eqref{eq:wrho_prod_removed}. As a result,
  \begin{equation}
    \norm{\mathcal{T}_\eta[\rho] - \Fo \rho \Fo^*}
    \leqslant \frac{1}{\pi} \norm{W_{\Fo \rho \Fo^*} -
      W_{\Fo \rho \Fo^*} *_1 g_\frac{1}{\eta}}_{\mathrm{L}^1}.
  \end{equation}
  By definition, \(G_\eta(x) = \eta^{-1} G_{1}(\eta^{-1}x)\)
  and \(\int_\RR G_\eta = 1\), so that by \cite{wheeden_measure_2015},
  theorem 9.6, the net \((G_\eta)_{\eta \in \RR_+^*}\) forms an
  approximation to identity. As a result,
  \begin{equation}
    \norm{W_{\Fo \rho \Fo^*} - W_{\Fo \rho \Fo^*} *_1 g_\frac{1}{\eta}}_{\mathrm{L}^1}
    \to 0 \qas \eta \to +\infty,
  \end{equation}
  and it follows that
  \begin{equation}
    \lim_{\eta \to \infty} \mathcal{T}_\eta[\rho] = \Fo \rho \Fo^*,
    \label{eq:convergence_pure_L1}
  \end{equation}
  for any such \(\rho\).
  Since the trace-class norm agrees with the Hilbert space norm for pure states,
  and \(\mathrm{L}^1(\RR) \cap \LR\) is dense in \(\LR\), equation
  \eqref{eq:convergence_pure_L1} holds by continuity of \(\mathcal{T}_\eta\) for
  any pure \(\rho \in D(\mathcal{H})\). Finally, since the set of finite convex sums
  of pure \(\rho \in D(\mathcal{H})\) is dense in \(D(\mathcal{H})\) (which is Banach) the
  result can be extended to mixed states.

  Reintroducing the edge-weight and measurement angles, we have
  \begin{equation}
    \lim_{\eta \to \infty} \mathcal{T}_\eta(\a,\b,\g,w)[\rho]
    = \So(w) \Fo \Uo(\a,\b,\g) \rho \Uo^*(\a,\b,\g) \Fo^* \So^*(w),
  \end{equation}
  as desired.
\end{proof}

\subsection{Proof of Proposistion \ref{prop:flow_convergence}}
\label{sapp:flow_convergence}

\FlowCircuit*

\begin{proof}
  As explained in section \ref{sec:spt}, we can decompose \(\Fch_\eta\) as a
  set of parallel paths with mediating edges.
  Each of these parallel paths corresponds to a sequence of single gate
  teleportations.
  All we need to worry about is ordering the mediating edges such that they
  appear before any teleportation of a node they are connected to. 
  This is possible since such an ordering exists if and only if there is a
  causal flow \cite{miyazaki_analysis_2015}.
  Then, by proposition \ref{prop:tele-circuit} each teleportation converges, and
  and since each gate teleporation channel is continuous it preserves limits.
\end{proof}

\subsection{Proof of Lemma \ref{lem:cont-stab}}
\label{sapp:approx_stab}

\ApproxStab*

\begin{proof}
  Let \(\phi \in \Hcal\) be normalised and Schwartz, pick \(s \in \RR\), and consider
  \begin{equation}
    \CXo_{1,2}(s) [\phi \otimes g_\eta] (x,y) = \exp(is\Qo_1\Po_1) [\phi \otimes g_\eta] (x,y)
                                                    = \phi(q) g_\eta(y + sx).
  \end{equation}
  Now, since \(\phi\) is square-integrable of norm \(1\), for any \(\e > 0\)
  there is some bounded measurable subset \(E \subseteq \RR\) such that
  \begin{equation}
    \int_{x \in E^\mathsf{c}} \abs{\phi(x)}^2 < \e,
  \end{equation}
  and
  \begin{align}
    \norm{\CXo_{1,2}(s) \phi \otimes g_\eta - \phi \otimes g_\eta}^2
    &= \frac{1}{\sqrt{\pi \eta^2}} \smashoperator{ \int_{\RR} } \smashoperator{ \int_{\RR} }
      \abs{\phi(x) e^{-\frac{(y+sx)^2}{2\eta^2}} - \phi(x) e^{-\frac{y^2}{2\eta^2}}}^2 \dd{x}\dd{y} \\
    &\leqslant \frac{1}{\sqrt{\pi \eta^2}} \smashoperator{ \int_{\RR} } \smashoperator{ \int_{\RR} }
      \abs{\phi(x)}^2 \abs{ e^{-\frac{(y+sx)^2}{2\eta^2}} - e^{-\frac{y^2}{2\eta^2}}}^2 \dd{x}\dd{y}\\
    &\leqslant \frac{1}{\sqrt{\pi \eta^2}} \smashoperator{ \int_{E} } \smashoperator{ \int_{\RR} }
      \abs{\phi(x)}^2 \abs{ e^{-\frac{(y+sx)^2}{2\eta^2}} - e^{-\frac{y^2}{2\eta^2}}}^2 \dd{y}\dd{x} \\
    &\hspace{1cm}+ \frac{1}{\sqrt{\pi \eta^2}} \smashoperator{ \int_{E^\mathsf{c}} } \smashoperator{ \int_{\RR} }
      \abs{\phi(x)}^2 \abs{ e^{-\frac{(y+sx)^2}{2\eta^2}} - e^{-\frac{y^2}{2\eta^2}}}^2 \dd{y}\dd{x} \\
    &\leqslant \frac{1}{\sqrt{\pi\eta^2}} \smashoperator{ \int_{x\in E} } \abs{\phi(x)}^2  \smashoperator{ \int_{\RR} }
      \abs{ e^{-\frac{(y+sx)^2}{2\eta^2}} - e^{-\frac{y^2}{2\eta^2}}}^2 \dd{y}\dd{x}
    + 2 \smashoperator{ \int_{E^\mathsf{c}} } 
      \abs{\phi(x)}^2 \dd{x}. 
  \end{align}

  Furthermore, using  for any \(x,y \in \RR\),
  \begin{equation}
    \abs{ e^{-\frac{(y+sx)^2}{2\eta^2}} - e^{-\frac{y^2}{2\eta^2}}}
    \leqslant \abs{sx} \cdot \max_{t \in \RR} \dv{t}(e^{-\frac{t^2}{2\eta^2}})
    \leqslant \frac{A}{\eta} \abs{sx}
  \end{equation}
  where \(A = \max_{t \in \RR} \dv{t}(e^{-\frac{t^2}{2}})\).
  Then,
  \begin{equation}
    \norm{\CXo_{1,2}(s) \phi \otimes g_\eta - \phi \otimes g_\eta}^2
    \leqslant \frac{A^2 s^2}{\eta^3\sqrt{\pi}} \smashoperator{ \int_{E} } \abs{x\phi(x)}^2 \dd{x} 
      + 2 \e, 
  \end{equation}
  and since \(\phi\) is Schwartz, \(\int_{x\in E} \abs{x\phi(x)}^2\) is bounded
  by some \(B > 0\).
  Finally,
  \begin{equation}
    \norm{\CXo_{1,2}(s) \phi \otimes g_\eta - \phi \otimes g_\eta}^2
    \leqslant \frac{A^2 B s^2}{\eta^3\sqrt{\pi}} 
      + 2 \e, 
  \end{equation}
  whence picking \(\eta > \sqrt[3]{\frac{A^2 B s^2}{\varepsilon \sqrt{\pi}}}\)
  we have \(\norm{\CXo_{1,2}(s) \phi \otimes g_\eta - \phi \otimes g_\eta}^2 <
  3\e\), and because \(\e > 0\) was arbitrary,
  \begin{equation}
    \lim_{\eta \to \infty} \norm{\CXo_{1,2}(s) \phi \otimes g_\eta - \phi \otimes g_\eta}^2 = 0.
  \end{equation}

  As every controlled stabiliser can be reduced to this case by commuting though
  \(\Eo_G\):
  \begin{equation}
    \CXo_{j,k}(s) \CZo_{j,N(k)}(s) \Eo_G =  \Eo_G \CXo_{j,k}(s),
  \end{equation}
  we can reduce the lemma to just this subcase, and we are done.
\end{proof}

\subsection{Proof of Theorem \ref{thm:CV-circuit}}
\label{app:CV_circuit}

\CVconvergence*

\begin{proof}
  By lemma \ref{lem:cv_cover} we obtain a graph \(G'\) that is approximately
  equivalent to \(G\) up to \(\CXo\) gates.
  Let \(\Eo_G^{(k)}\) be the product of \(\CZo\) gates in \(G'\) from layer
  \(V_k\) into its outputs and \(\To^{(k)}\) the \(\CXo\) gates obtained from
  the corresponing triangularisation procedure.
  By lemmas \ref{lem:triangularisation} and \ref{lem:cv_cover} for any
  \(\mathtt{A} > 0\) we have, for high enough squeezing, that
  \begin{equation}
    \norm{\Eo_G g_\eta^{\otimes \abs{I^{\mathsf{c}}}} \otimes \phi
      - \prod_{k = 1}^n \Big(\To^{(k)} \Eo_G^{(k)}\Big)
       g_\eta^{\otimes \abs{I^{\mathsf{c}}}} \otimes \phi} < \mathtt{A}.
  \end{equation}

  Now, none of the edges in \(\To^{(k)}\) \(\Eo_G^{(k)}\) for \(k < n\) touch the
  nodes in \(V_{k+1}\), so that we can bring the auxiliary squeezed states
  \(\ket{\eta}_v\) for \(v \in V_k\) forward until \(\Eo_G^{(k)}\).
  Since there is a causal flow \(V_{k+1} \to V_k\),
  \begin{equation}
    \prod_{k=1}^n \Big(\To^{(k)} \Eo_G^{(k)}\Big)
    g_\eta^{\otimes \abs{I^{\mathsf{c}}}} \otimes \phi
    = \To^{(k)} \circ \Och^{(k)} \circ \cdots \circ \To^{(1)} \circ \Och^{(1)}
      \big[ \phi \big]
  \end{equation}
  where \(\Och^{(k)}\) is the channel associated to the causal flow procedure
  \(V_{k+1} \to V_k\).
  Then by proposition \ref{prop:flow_convergence}, we can perform an SPT for
  each \(\Och^{(k)}\), and
  \begin{equation}
    \lim_{\eta \to \infty}
    \Cch_\eta(\vec{\a}, \vec{\b}, \vec{\g})
    = \prod_{k=1}^{n} \To^{(k)} \; \Uo_{SPT}^{(k)}(\vec{\a},\vec{\b},\vec{\g}),
  \end{equation}
  by continuity.
\end{proof}

\section{A polynomial time algorithm for finding CV-flows}
\label{app:algorithm}
The following pseudocode algorithm finds a CV-flow for an open graph if it has
one (we interpret the return value \(\varnothing\) as the graph not having a
flow).
Our algorithm is based on \cite{mhalla_finding_2008}, which contains an almost identical algorthm for the qubit case.

\begin{algorithmic}[1]
  \State \textbf{input:} A CV open graph
  \State \textbf{output:} A CV-flow
  \Statex
  \Procedure{CV-Flow}{$G, I, O$}
    \ForAll{\(v \in O\)}
      \State \(\mathrm{layer}(v) \coloneqq 0\)
    \EndFor
    \State  \textbf{return} \Call{CV-Flow-aux}{\(G, I, O\), \textrm{layer},
      \textrm{flow}, 1}
  \EndProcedure
  \Statex
  \Procedure{CV-Flow-aux}{\(G\), In, Out, \textrm{layer}, $k$}
    \State \(O' \coloneqq O \setminus I\) \Comment{Nodes onto which we can correct}
    \State \(C \coloneqq \varnothing\) \Comment{Nodes which we are correcting in
    this layer}
    \ForAll{\(v \in G \setminus O\)}
      \State \textbf{solve} in \(\RR\): \(A_C \vec{c} = 1_{\{v\}}\) assuming \(v
      \prec o\) for all \(v \in O^\mathsf{c}, o \in O\)
      \Comment{\(A_C\) is the correction matrix of \(C\)}
      \If{there is a solution \(\vec{c}\)}
        \State \(C \coloneqq C \cup \{v\}\)
        \State \(\mathrm{layer}(v) \coloneqq k\) \Comment{Assign \(v\) to layer \(k\)}
        \State \(\mathrm{flow}(k) = \vec{c}\) \Comment{The corrections for layer
          \(k\)}
      \EndIf
    \EndFor
    \If{\(C = \varnothing\)} \Comment{If we can no longer correct for additional
    nodes, either:}
      \If{\(O = G\)} 
        \State \textbf{return} \((\mathrm{flow}, \mathrm{layer})\) \Comment{we have found a CV-flow for the whole open graph; or,}
      \Else 
        \State \textbf{return} \(\varnothing\)\Comment{there is no CV-flow.}
      \EndIf
    \Else
      \State  \textbf{return} \Call{CV-Flow-aux}{\(G, I, O \cup C\),
        \textrm{layer}, \textrm{flow}, $k+1$}
    \EndIf
  \EndProcedure
\end{algorithmic}

The return value is a pair of functions, \(\mathrm{layer} : G \to \NN\) which
assigns each node to a layer, and \(\mathrm{flow} : \NN \to \RR^N\) which
returns the correction coefficients for each layer.
The ordering for the CV-flow is implicitely given by the layers: all nodes within
layers are unordered with respecting to each other and nodes in layer \(k+1\)
are less than nodes in layer \(k\).

For a proof of the running time of this algorithm, we refer the reader to
\cite{mhalla_finding_2008}.
The only real difference between that algorithm and ours are the lines \(13\) to
\(18\): the correction equations are solved over \(\RR\) instead of \(\ZZ_2\),
and the type of the solution is subtly different.

\section{Comparing g-flow and CV-flow}
\label{app:comparison}
\begin{Def}
  An open graph \((G, I, O)\) has a \textbf{generalised flow}, or
  gflow, if there exists a map \(g: I^\mathsf{c} \to
  \mathcal{P}(O^\mathsf{c})\) and a partial order \(\prec\) over \(G\) such that
  for all \(i \in I^\mathsf{c}\):
  \begin{itemize}
  \item if \(j \in g(i)\) and \(i \neq j\) then \(i \prec j\);
  \item if \(j \prec i\) then \(j \notin \mathrm{Odd}(g(i))\).\qedhere
  \end{itemize}
  \vspace{-3mm}
  \label{def:gflow}
\end{Def}

It is not immediately clear if gflow and CV-flow are equivalent properties or
not, or even if one is strictly stronger than the other.
We construct counterexamples to either implication, showing that these are
indeed entirely independant properties.
Thus, a graph can have both (as in the case of flow), either or neither.

\begin{Prop}
  The open graph
  \begin{center}
    \begin{tikzpicture}
	\begin{pgfonlayer}{nodelayer}
		\node [style={white_circle}] (2) at (-1.75, 1) {};
		\node [style={white_circle}] (3) at (1.75, 1) {};
		\node [style={black_circle}] (4) at (-1.75, -1) {};
		\node [style={black_circle}] (5) at (0, 2) {};
		\node [style={white_circle}] (6) at (0, -2) {};
		\node [style={black_circle}] (7) at (1.75, -1) {};
		\node [style={black_circle}] (8) at (0, 0) {};
	\end{pgfonlayer}
	\begin{pgfonlayer}{edgelayer}
		\draw (2) to (4);
		\draw (4) to (6);
		\draw (6) to (7);
		\draw (7) to (3);
		\draw (5) to (2);
		\draw (2) to (8);
		\draw (8) to (3);
		\draw (5) to (8);
		\draw (8) to (6);
		\draw (5) to (3);
	\end{pgfonlayer}
\end{tikzpicture}

  \end{center}
  has a gflow but no CV-flow.
\end{Prop}

\begin{proof}
  The graph has a gflow, given (for example) by the measurement order
  \begin{center}
    \begin{tikzpicture}
	\begin{pgfonlayer}{nodelayer}
		\node [style={white_circle}] (2) at (-1.75, 1) {};
		\node [style={white_circle}] (3) at (1.75, 1) {};
		\node [style={black_circle}] (4) at (-1.75, -1) {};
		\node [style={black_circle}] (5) at (0, 2) {};
		\node [style={white_circle}] (6) at (0, -2) {};
		\node [style={black_circle}] (7) at (1.75, -1) {};
		\node [style={black_circle}] (8) at (0, 0) {};
		\node [style=none] (9) at (-2.5, -1.25) {1};
		\node [style=none] (10) at (2.5, -1.25) {2};
		\node [style=none] (11) at (0, 2.75) {3};
		\node [style=none] (12) at (0.5, -0.25) {4};
	\end{pgfonlayer}
	\begin{pgfonlayer}{edgelayer}
		\draw (2) to (4);
		\draw (4) to (6);
		\draw (6) to (7);
		\draw (7) to (3);
		\draw (3) to (5);
		\draw (5) to (2);
		\draw (2) to (8);
		\draw (8) to (3);
		\draw (5) to (8);
		\draw (8) to (6);
	\end{pgfonlayer}
\end{tikzpicture}
,
  \end{center}
  where we correct each of the first three nodes onto an unmeasured neighbour,
  and the fourth node is corrected on all three outputs.
  The modularity ensures that this last correction has no backaction on nodes
  \(1\) to \(3\).

  Now, we show that the graph does not have CV-flow, by showing that no node in
  the graph can be measured last. There are two cases: either we measure the
  central node last (case 1), or we measure one of the outer nodes last (case 2)
  (by symmetry, these are all equivalent).

  \paragraph{Case 1:} The maximal correction subgraph, where we are correcting a
  measurement on the blue node and the black nodes have already been
  measured, is given by
  \begin{center}
    \begin{tikzpicture}
	\begin{pgfonlayer}{nodelayer}
		\node [style={white_circle}] (2) at (-1.75, 1) {};
		\node [style={white_circle}] (3) at (1.75, 1) {};
		\node [style={black_circle}] (4) at (-1.75, -1) {};
		\node [style={black_circle}] (5) at (0, 2) {};
		\node [style={white_circle}] (6) at (0, -2) {};
		\node [style={black_circle}] (7) at (1.75, -1) {};
		\node [style={blue_circle}] (8) at (0, 0) {};
	\end{pgfonlayer}
	\begin{pgfonlayer}{edgelayer}
		\draw (2) to (4);
		\draw (4) to (6);
		\draw (6) to (7);
		\draw (7) to (3);
		\draw (5) to (2);
		\draw (2) to (8);
		\draw (8) to (3);
		\draw (8) to (6);
		\draw (5) to (3);
	\end{pgfonlayer}
\end{tikzpicture}
,\hspace{5mm} for which the correction equation
    has superior matrix \hspace{5mm}\(\left(
      \begin{array}{@{}ccc|c@{}}
        1 & 1 & 1 & 1 \\
        1 & 1 & 0 & 0 \\
        1 & 0 & 1 & 0 \\
        0 & 1 & 1 & 0 \\
      \end{array}\right).
    \)
  \end{center}

  \paragraph{Case 2:} The maximal correction subgraph is given by
  \begin{center}
    \begin{tikzpicture}
	\begin{pgfonlayer}{nodelayer}
		\node [style={white_circle}] (2) at (-1.75, 1) {};
		\node [style={white_circle}] (3) at (1.75, 1) {};
		\node [style={black_circle}] (4) at (-1.75, -1) {};
		\node [style={black_circle}] (5) at (0, 2) {};
		\node [style={white_circle}] (6) at (0, -2) {};
		\node [style={blue_circle}] (7) at (1.75, -1) {};
		\node [style={black_circle}] (8) at (0, 0) {};
	\end{pgfonlayer}
	\begin{pgfonlayer}{edgelayer}
		\draw (2) to (4);
		\draw (4) to (6);
		\draw (6) to (7);
		\draw (7) to (3);
		\draw (5) to (2);
		\draw (2) to (8);
		\draw (8) to (3);
		\draw (8) to (6);
		\draw (5) to (3);
	\end{pgfonlayer}
\end{tikzpicture}
,\hspace{5mm} for which the correction equation
    has superior matrix \hspace{5mm}\(\left(
      \begin{array}{@{}ccc|c@{}}
        1 & 1 & 1 & 0 \\
        1 & 1 & 0 & 1 \\
        1 & 0 & 1 & 0 \\
        0 & 1 & 1 & 0 \\
      \end{array}\right).
    \)
  \end{center}
  Since both of these superior matrices have reduced row echelon form
  \begin{equation*}
    \left(\begin{array}{@{}ccc|c@{}}
            1 & 0 & 0 & 0 \\
            0 & 1 & 0 & 0 \\
            0 & 0 & 1 & 0 \\
            0 & 0 & 0 & 1 \\
          \end{array}\right),
  \end{equation*}
  equation the correction equation has no solution in either case. For any CV-flow
  to exist, we must be able to measure and correct some node of the graph last
  (every partial order extends consistently to a total order), so we have
  constructed a graph that has a gflow but no CV-flow.
\end{proof}

\begin{Prop}
  The open graph
  \begin{center}
    \begin{tikzpicture}
	\begin{pgfonlayer}{nodelayer}
		\node [style={white_circle}] (2) at (-1.75, 1) {};
		\node [style={white_circle}] (3) at (1.75, 1) {};
		\node [style={black_circle}] (4) at (-1.75, -1) {};
		\node [style={black_circle}] (5) at (0, 2) {};
		\node [style={white_circle}] (6) at (0, -2) {};
		\node [style={black_circle}] (7) at (1.75, -1) {};
	\end{pgfonlayer}
	\begin{pgfonlayer}{edgelayer}
		\draw (2) to (4);
		\draw (4) to (6);
		\draw (6) to (7);
		\draw (7) to (3);
		\draw (5) to (2);
		\draw (5) to (3);
	\end{pgfonlayer}
\end{tikzpicture}

  \end{center}
  has a CV-flow but no gflow.
  \label{prop:CV_flow_graph}
\end{Prop}

\begin{proof}
  Any measurement order on the nodes is equivalent by symmetry, so the graph has
  a unique CV-flow (up to rotations) given by:
  \begin{center}
    \begin{tikzpicture}
	\begin{pgfonlayer}{nodelayer}
		\node [style={white_circle}] (2) at (-1.75, 1) {};
		\node [style={white_circle}] (3) at (1.75, 1) {};
		\node [style={black_circle}] (4) at (-1.75, -1) {};
		\node [style={black_circle}] (5) at (0, 2) {};
		\node [style={white_circle}] (6) at (0, -2) {};
		\node [style={black_circle}] (7) at (1.75, -1) {};
		\node [style=none] (9) at (-2.5, -1.25) {1};
		\node [style=none] (10) at (2.5, -1.25) {2};
		\node [style=none] (11) at (0, 2.75) {3};
	\end{pgfonlayer}
	\begin{pgfonlayer}{edgelayer}
		\draw (2) to (4);
		\draw (4) to (6);
		\draw (6) to (7);
		\draw (7) to (3);
		\draw (3) to (5);
		\draw (5) to (2);
	\end{pgfonlayer}
\end{tikzpicture}
,
  \end{center}
  where we correct the first two measurement onto a neighbouring output (as per
  flow), and the last measurement outcome \(m\) as:
  \begin{center}
    \begin{tikzpicture}
	\begin{pgfonlayer}{nodelayer}
		\node [style={white_circle}] (2) at (-1.75, 0.5) {};
		\node [style={white_circle}] (3) at (1.75, 0.5) {};
		\node [style={black_circle}] (4) at (-1.75, -1.5) {};
		\node [style={blue_circle}] (5) at (0, 1.5) {};
		\node [style={white_circle}] (6) at (0, -2.5) {};
		\node [style={black_circle}] (7) at (1.75, -1.5) {};
		\node [style=none] (8) at (-3, 1.25) {\(\mathsf{X}(\frac{m}{2})\)};
		\node [style=none] (9) at (3, 1.25) {\(\mathsf{X}(\frac{m}{2})\)};
		\node [style=none] (10) at (0, -3.5) {\(\mathsf{X}(-\frac{m}{2})\)};
		\node [style=blank] (11) at (0, -0.75) {};
		\node [style=none] (12) at (0, 2.5) {\({\color{blue}\mathsf{Z}(m)}\)};
	\end{pgfonlayer}
	\begin{pgfonlayer}{edgelayer}
		\draw (2) to (4);
		\draw (4) to (6);
		\draw (6) to (7);
		\draw (7) to (3);
		\draw (5) to (2);
		\draw (5) to (3);
	\end{pgfonlayer}
\end{tikzpicture}
.
  \end{center}

  By definition~\ref{def:gflow}, this last correction is not possible in
  DV,  as every possible subset of the unmeasured nodes (in white) that connects
  oddly to the measured node (in blue) also connects oddly to one of the
  previously corrected nodes (in black).
\end{proof}

\section{Depth complexity advantage in CV-MBQC}
\label{app:depth}
One of the main reasons for interest in measurement-based quantum computation
originally came from the reduction in depth of certain quantum computations when
the corresponding quantum circuit is reformulated as an MBQC
\cite{broadbent_parallelizing_2009, miyazaki_analysis_2015}. This improvement
does not come for free: to obtain an MBQC that implements the same computation
typically requires one to increase the number of quantum systems involved in the
computation. We show that these results straightforwardly transfer to the CV
setting, by constructing a computation which takes at least logarithmic depth as
a quantum circuit, but can be implemented in constant depth as an MBQC.

This result was originally obtained in the qubit case in
\cite{broadbent_parallelizing_2009}, by considering a simple arithmetic problem.
The problem asks one to sum \(N\) binary digits, and a simple argument shows
that accessing all the binary digits requires at least a logarithmic number of
sequential binary gates. We reformulate this problem in CV. In order to avoid
the questions of precision when encoding arbitrary real numbers, we instead
consider the analogous problem for integers.

Each integer in the input (to be summed) is encoded in the position observable
of a qumode, use the following state: for any \(k \in \ZZ\), let \(f_k \in
\mathrm{L}^2(\RR)\) be given by
\begin{equation}
  f_k(x) \coloneqq
  \begin{cases}
    1 \qif x \in [k,k+1); \\
    0 \qq{otherwise.}
  \end{cases}
\end{equation}
The point is that, given two such states \(f_m\) and \(f_n\) (where \(m,n \in
\ZZ\)), we can use a \(\CZo\) gate to obtain the sum of \(m\) and \(n\), itself
encoded in a state of the form \(f_{m+n}\):
\begin{center}
  \leavevmode
  \centering
  \Qcircuit @C=1.2em @R=1.2em {
    \lstick{f_m} & \gate{\Fo} & \ctrl{1} & \qw & \qw \\
    \lstick{f_n} & \qw & \control \qw & \gate{\Fo} & \rstick{\Fo f_{m+n}} \qw \\           
  }
  \vspace{3mm}
\end{center}
We ignore the first output of the circuit.

Now, let \((k_j)_{j=1}^N \in \ZZ^N\) be a set of integers corresponding to the
input of the problem, and consider the CV-MBQC given by the open graph
\begin{center}
  \scalebox{1.05}{\begin{tikzpicture}
	\begin{pgfonlayer}{nodelayer}
		\node [style={black_square}] (0) at (-8, 1) {};
		\node [style={black_square}] (1) at (-6, 1) {};
		\node [style={black_square}] (2) at (-4, 1) {};
		\node [style={black_square}] (3) at (-2, 1) {};
		\node [style={black_square}] (5) at (2, 1) {};
		\node [style={black_square}] (6) at (4, 1) {};
		\node [style={black_square}] (7) at (6, 1) {};
		\node [style={white_circle}] (8) at (-8, -1) {};
		\node [style={white_circle}] (9) at (-6, -1) {};
		\node [style={white_circle}] (10) at (-4, -1) {};
		\node [style={white_circle}] (11) at (-2, -1) {};
		\node [style={white_circle}] (13) at (2, -1) {};
		\node [style={white_circle}] (14) at (4, -1) {};
		\node [style={white_circle}] (15) at (6, -1) {};
		\node [style=none] (16) at (-1, 0) {};
		\node [style=none] (17) at (1, 0) {};
		\node [style=none] (18) at (0, 0) {\(\cdots\)};
		\node [style=none] (19) at (-8, 1.75) {\(I_1\)};
		\node [style=none] (20) at (-6, 1.75) {\(I_2\)};
		\node [style=none] (21) at (-4, 1.75) {\(I_3\)};
		\node [style=none] (22) at (-2, 1.75) {\(I_4\)};
		\node [style=none] (23) at (2, 1.75) {\(I_{N-2}\)};
		\node [style=none] (24) at (4, 1.75) {\(I_{N-1}\)};
		\node [style=none] (25) at (6, 1.75) {\(I_N\)};
		\node [style=none] (33) at (-8, -1.75) {\(O_1\)};
		\node [style=none] (34) at (-6, -1.75) {\(O_2\)};
		\node [style=none] (35) at (-4, -1.75) {\(O_3\)};
		\node [style=none] (36) at (-2, -1.75) {\(O_4\)};
		\node [style=none] (37) at (2, -1.75) {\(O_{N-2}\)};
		\node [style=none] (38) at (4, -1.75) {\(O_{N-1}\)};
		\node [style=none] (39) at (6, -1.75) {\(O_N\)};
	\end{pgfonlayer}
	\begin{pgfonlayer}{edgelayer}
		\draw (0) to (8);
		\draw (1) to (9);
		\draw (2) to (10);
		\draw (3) to (11);
		\draw (5) to (13);
		\draw (6) to (14);
		\draw (7) to (15);
		\draw (1) to (8);
		\draw (9) to (2);
		\draw (10) to (3);
		\draw (11) to (16.center);
		\draw (17.center) to (5);
		\draw (13) to (6);
		\draw (14) to (7);
	\end{pgfonlayer}
\end{tikzpicture}}
\end{center}
where the input \(I_n\) is prepared in the state \(f_{k_n}\).
This open graph has a 1-step CV-flow where all the inputs are in the layer \(I
\prec O\).
The correction equation is given by
\begin{equation}
  \begin{pmatrix}
    1 & 0 & 0 & 0 & 0 & \cdots & 0 & 0 \\
    1 & 1 & 0 & 0 & 0 & \cdots & 0 & 0 \\
    0 & 1 & 1 & 0 & 0 & \cdots & 0 & 0 \\
    0 & 0 & 1 & 1 & 0 & \cdots & 0 & 0 \\
    0 & 0 & 0 & 1 & 1 & \cdots & 0 & 0 \\
    \vdots & \vdots & \vdots & \vdots & \vdots & \ddots & \vdots & \vdots \\
    0 & 0 & 0 & 0 & 0 & \cdots & 1 & 0\\
    0 & 0 & 0 & 0 & 0 & \cdots & 1 & 1
  \end{pmatrix} \vec{c} = \mqty(\xmat*{m}{5}{1} \\ \vdots \\ m_{N-1} \\ m_{N}),
\end{equation}
which has a solution \(\vec{c} \in \RR^N\) whose \(n\)-th element is \(c_n =
\sum_{j=1}^n m_j\).
Performing the circuit extraction as described in section \ref{sec:circuit} and
taking \(\a = \b = \g = 0\) at each measured node (that is, we just measure
\(\Po\)), we see that the MBQC is equivalent in the infinite-squeezing limit to
the circuit
\begin{center}
  \leavevmode
  \centering
  \Qcircuit @C=1.2em @R=1.2em {
    \lstick{f_{k_1}} & \gate{\Fo} & \ctrl{1} & \qw & \qw          & \qw & \qw & \qw & \qw & \qw & \qw & \qw\\
    \lstick{f_{k_2}} & \qw & \control \qw & \gate{\Fo} & \ctrl{1} & \qw & \qw & \qw & \qw & \qw & \qw & \qw\\           
    \lstick{f_{k_3}} & \qw & \qw & \qw & \control \qw & \gate{\Fo} & \ctrl{1}
    \qw & \qw & \qw & \qw & \qw & \qw \\
    & & & & & & & & & \\
    & & & & & \vdots & & & & \\
    & & & & & & & & & \\
    \lstick{f_{k_{N-1}}} & \qw & \qw & \qw & \qw & \qw & \qw & \ctrl{-1} \qw &
    \gate{\Fo} & \ctrl{1} & \qw & \qw \\
    \lstick{f_{k_{N}}}   & \qw & \qw & \qw & \qw & \qw & \qw & \qw & \qw &
    \control \qw & \gate{\Fo} \qw & \rstick{\Fo f_{\sum_{j=1}^N k_j}} \qw \\ 
  }
  \vspace{3mm}
\end{center}
where we only care about the last output.
Since the graph state can be generated in two steps provided access to the input
states (generating the auxiliary squeezed states then performing the entangling
unitaries), the CV-MBQC yields a quantum circuit for generating the output state
\(\Fo f_{\sum_{j=1}^N k_j}\) in constant depth and with using \(2N\) qumodes. In
particular, this allows us to determine \(\sum_{j=1}^N k_j\) for any input
\((k_j)_{j=1}^N \in \ZZ^N\), since performing a momentum measurement on the
output state \(\Fo f_{\sum_{j=1}^N k_j}\) is guaranteed to return a value in
\([\sum_{j=1}^N k_j,\sum_{j=1}^N k_j+1)\).

Now, from  \cite{broadbent_parallelizing_2009} we have:
\begin{Prop}
  Let \(\So\) a unitary operator on \(\mathcal{H}\) that restricts to the map
  \begin{equation}
    \So (\bigotimes_{j=1}^N f_{k_j}) = \bigotimes_{j=1}^{N-1} f_{k_j} \otimes \left( f_{\sum_{j=1}^N k_j} \right)
  \end{equation}
  on the subspace of \(\mathcal{H}\) generated by the \(f_k\) for \(k \in \ZZ\).
  Then any circuit of \(1\) and \(2\)-qumode gates that implements \(\So\) has
  depth in \(\Omega(\log_2(N))\).
\end{Prop}

This result follows from an argument on the ``backwards light-cone'' of the
outputs.
In general, the output of a \(k\)-ary gate can depend only on the \(k\) inputs,
so that composing \(n\) times, the output can depend on at most \(k^n\) inputs.
As a result, the number of gates needed to treat \(N\) inputs is at least
\(\log_k(N)\).

Thus, this example demonstrates a clear depth advantage for CV-MBQC: at the cost
of doubling the width of the circuit, the unitary \(\So\) is implemented in
constant rather than logarithmic depth.

\end{document}